\documentclass[showpacs,twocolumn,floatfix,superscriptaddress,notitlepage]{revtex4-2}
\usepackage{graphicx}
\usepackage{dcolumn}
\usepackage{bm}
\usepackage[colorlinks=true,urlcolor=black,citecolor=blue,linkcolor=blue]{hyperref}
\usepackage{suffix}
\usepackage{xurl}

\usepackage{xr}
\usepackage{multirow}
\usepackage{minitoc}
\usepackage{xfrac}
\usepackage{amsmath}
\usepackage{amssymb}
\usepackage{braket}
\usepackage{physics}
\usepackage{amsthm}
\usepackage{natbib}
\usepackage{mathtools}
\usepackage{diagbox}
\usepackage[utf8]{inputenc}
\usepackage[english]{babel}
\usepackage{scalerel}[2014/03/10]
\usepackage{stackengine}
\usepackage{algorithm}
\usepackage{algpseudocode}
\usepackage{tikz}
\usetikzlibrary{quantikz}
\usepackage{caption}
\usepackage{subcaption}

\usepackage{ragged2e}

\newtheorem{theorem}{Theorem}[]

\newtheorem{lemma}[]{Lemma}

\theoremstyle{definition}

\makeatletter
\@addtoreset{proofpart}{theorem}
\@addtoreset{proofcase}{theorem}
\makeatother
\renewcommand{\exp}[1]{e^{ #1 }}

\newcommand{\comment}[1]{{}}

\newcommand{\thu}{Department of Mathematics, Tsinghua University,  Beijing 100084, China}
\newcommand{\YMSC}{Yau Mathematical Sciences Center, Tsinghua University,  Beijing 100084, China}
\newcommand{\bimsa}{Yanqi Lake Beijing Institute of Mathematical Sciences and Applications, Beijing 100407, China }

\begin{document}

\title{Clifford Perturbation Approximation for Quantum Error Mitigation}

\author{Ruiqi Zhang}
\thanks{These authors contributed equally to this work.}
\affiliation{\YMSC}
\affiliation{\thu}

\author{Yuguo Shao}
\thanks{These authors contributed equally to this work.}
\affiliation{\YMSC}
\affiliation{\thu}

\author{Fuchuan Wei}
\affiliation{\YMSC}
\affiliation{\thu}

\author{Song Cheng}
\thanks{chengsong@bimsa.cn}
\affiliation{\bimsa}

\author{Zhaohui Wei}
\thanks{weizhaohui@gmail.com}
\affiliation{\YMSC}
\affiliation{\bimsa}

\author{Zhengwei Liu}
\thanks{liuzhengwei@mail.tsinghua.edu.cn}
\affiliation{\YMSC}
\affiliation{\thu}
\affiliation{\bimsa}


\begin{abstract}

Quantum error mitigation (QEM) is critical for harnessing the potential of near-term quantum devices. Particularly, QEM protocols can be designed based on machine learning, where the mapping between noisy computational outputs and ideal ones can be learned on a training set consisting of Clifford circuits or near-Clifford circuits that contain only a limited number of non-Clifford gates. This learned mapping is then applied to noisy target circuits to estimate the ideal computational output. 
In this work, we propose a learning-based error mitigation framework called Clifford Perturbation Data Regression (CPDR), which constructs training sets by Clifford circuits with small perturbations. 
Specifically, these circuits are parameterized quantum circuits, where the rotation angles of the gates are restricted to a narrow range, ensuring that the gates remain close to Clifford gates. This design enables the efficient simulation of the training circuits using the Sparse Pauli Dynamics method. 
As a result, CPDR is able to utilize training sets with a better diversity to train the model, compared with previous learning-based QEM models that construct training sets with only Clifford or near-Clifford circuits.
Numerical simulations on small-scale Ising model circuits demonstrate that the performance of CPDR dramatically outperforms that of existing methods such as Zero-Noise Extrapolation and learning-based Probabilistic Error Cancellation. Furthermore, using the experimental data from IBM’s 127-qubit Eagle processor, our findings suggest that CPDR demonstrates improved accuracy compared to the original mitigation results reported in [Nature \textbf{618}, 500 (2023)].

\end{abstract}

\maketitle


\section{Introduction}
\label{sec:Introduction}

Quantum computing is approaching a pivotal milestone, at which its advantages over classical computing may be realized across various practical computational tasks~\cite{shor1994algorithms,lloyd1996universal, harrow2009quantum, shor1999polynomial}.
However, noise is unavoidably introduced in quantum devices, presenting a substantial challenge to their practical deployment~\cite{wang2021noise,sun2024sudden,shao2024simulating,aharonov2023polynomial,fontana2023classical}.
To address this issue and enable fault-tolerant quantum computation, \textit{Quantum Error Correction}~(QEC) presents a promising solution~\cite{girvin2021introduction,shor1995scheme,lidar2013quantum}.
Nevertheless, near-term quantum devices~\cite{preskill2018quantum,bharti2022noisy,chen2022complexity} fall far short of the requirements for fault-tolerance, due to their limited number of qubits and insufficient gate fidelities.
In order to suppress noise in near-term devices, \textit{Quantum Error Mitigation}~(QEM) techniques have emerged as a vital tool~\cite{qin2022overview, cai2023quantum, van2023probabilistic,kim2023evidence,guo2024experimental}. 
QEM seeks to recover the ideal expectation value of an observable on the output of a target quantum circuit, rather than the noiseless output quantum state. 
Up to now, quite a few mitigation protocols have been proposed~\cite{temme2017error,endo2018practical,huggins2021virtual,strikis2021learning,cai2021multi,liu2024virtual}, which include Zero-noise Extrapolation (ZNE) and Probabilistic Error Cancellation (PEC). Basically, ZNE mitigates the noise by running the same circuit at different noise levels and extrapolating the results to the zero-noise limit~\cite{temme2017error,li2017efficient}.
While PEC protocol works by inserting additional gates into noisy circuits and constructing the inverse mapping of each noise channel with a linear combination of these modified circuits~\cite{temme2017error,van2023probabilistic}.

With the popularity of machine learning in various fields, learning-based QEM protocols have been proposed~\cite{strikis2021learning,czarnik2021error,lowe2021unified,czarnik2022improving}. 
These protocols first learn the mapping between noisy and ideal expectation values using a training set, then apply this learned mapping to noisy target circuits to estimate the ideal expectation values.
The performance of learning-based QEM models depends heavily on the construction of the training set. 
The construction of the training set is typically based on the following two principles:
First, it is crucial to ensure that the training set consists of circuits that closely resemble the target circuits, so that the noise effects align with those of the target circuits.
Second, the circuits in the training set should be classically simulable to obtain ideal expectation values, which serve as noiseless references in the training data.
For example, the learning-based PEC method~\cite{strikis2021learning} aims to estimate the noiseless expectation value of a target circuit by using a linear combination of noisy expectation values from circuits with additional inserted gates.
The linear combination's coefficients are learned in the training set, in which circuits are created by substituting non-Clifford gates in the target circuit with various Clifford gates.
In previous works, the training sets have been primarily composed of circuits dominated by Clifford gates~\cite{czarnik2021error, lowe2021unified, strikis2021learning}, whose simulation can be efficiently conducted using the Gottesman-Knill algorithm~\cite{gottesman1998heisenberg,bravyi2016improved,howard2017application,aaronson2004improved}.

In recent years, simulation methods based on path integrals have made significant advances, thereby broadening the range of circuits that can be simulated~\cite{shao2024simulating,beguvsic2023simulating,gao2018efficient,aharonov2023polynomial,schuster2024polynomial,beguvsic2024real,angrisani2024classically,beguvsic2024fast,rudolph2023classical,fontana2023classical,bermejo2024quantum}.
Ref.~\cite{beguvsic2023simulating} proposes an efficient technique called \textit{Sparse Pauli Dynamics} (SPD), which simulates Clifford perturbation circuits, specifically quantum circuits with rotation gates whose angles lie within a predefined threshold range. Subsequent studies~\cite{beguvsic2024fast,beguvsic2024real} have numerically validated the accuracy of this approach in various practical circuits, even discarding the angle restriction. 

In this paper, we propose a novel learning-based mitigation framework, called \textit{Clifford Perturbation Data Regression} (CPDR).
This framework selects a training set from Clifford perturbation circuits and employs the SPD method to simulate noiseless expectation values for these circuits.
To ensure the reliability of the SPD method, we provide a theoretical guarantee for the accuracy with respect to the rotation angle threshold. This guarantees that the noiseless expectation values in the training set are precise.
Within this framework, we develop two protocols: CPDR-ZNE and CPDR-PEC, based on the ZNE and PEC methods, respectively.
Compared to existing approaches, applying the CPDR framework incurs no significant resource overhead while achieving improved accuracy in mitigating errors.
To benchmark our CPDR protocols, we numerically simulate their performances on small-scale Ising model evolution circuits, and then compare them with original ZNE and learning-based PEC. 
The results demonstrate that utilizing CPDR framework achieves superior accuracy.
Finally, we apply the CPDR-ZNE protocol to the experimental data from IBM's 127-qubit Eagle processor~\cite{kim2023}, demonstrating its ease of implementation and superior performance over the original mitigation methods applied in experiment~\cite{kim2023evidence}.

The remainder of this paper is organized as follows:
In Section II, we provide an overview of learning-based QEM. 
In Section III, we introduce the theoretical foundation and methodology of the numerical simulation method SPD, along with the learning-based error mitigation framework CPDR. 
Section IV presents numerical simulations of SPD, compares CPDR with common QEM protocols, and demonstrates the effectiveness of CPDR on IBM’s 127-qubit Eagle processor data, showcasing the practical benefits of our method on real-world quantum hardware.
Finally, we conclude our results and discuss future research directions in Section V.

\section{Preliminaries}

In the current NISQ era, parameterized quantum circuits (PQC) are widely used in many near-term algorithms~\cite{kandala2017hardware,farhi2014quantum}.
A typical $n$-qubit PQC, denoted as $\mathcal{C}(\bm{\theta})$, consists of a sequence of Pauli rotation gates and non-parameterized Clifford gates.
The Pauli rotation gates are represented as $e^{-i\frac{\theta}{2} P}$, where $P\in\{\mathbb{I},X,Y,Z\}^{\otimes n}$. The Clifford gates are the unitary operators that normalize the Pauli group $\bm{P}_n$: $\{C\in U_{2^n} \mid C\bm{P}_nC^\dagger=\bm{P}_n\}$. 
Without loss of generality, we assume that PQCs follow the form:
\begin{equation}\label{eq:parameterized_circuit}
  \mathcal{C}(\bm{\theta})=U_L(\theta_L)  \cdots {U}_1(\theta_1),
\end{equation}
where $\bm{\theta}=(\theta_1,\cdots,\theta_L)$ are rotation angles and $L$ is the depth of circuit. Each unitary ${U}_i(\theta_i):=\exp{-i \theta_i P_i / 2}C_i $ comprises a Clifford operator $C_i$ and a rotation $\exp{-i \theta_i P_i / 2}$ on Pauli operator $P_i\in\{\mathbb{I},X,Y,Z\}^{\otimes n}$ with angle $\theta_i$.

In this context, the quantum circuit $\mathcal{C}(\bm{\theta})$ is applied to an initial state $\rho$, and what we are interested in is the expectation value of an observable $O$, given by
\begin{equation}
  \langle O \rangle = \tr{O \mathcal{C}(\bm{\theta})\rho \mathcal{C}(\bm{\theta})^\dagger}.
\end{equation}
Suppose the noise in the quantum device is characterized by a noise rate $\lambda\in [0,1]$, and let $f(\mathcal{C}(\bm{\theta}),\lambda)$ be the expectation value measured by the noisy quantum device.
As described in Ref.~\cite{temme2017error}, by modeling the noise with the Lindblad master equation and expanding the system state into a Born series, the noisy expectation value can be approximated as a polynomial of $\lambda$ for any order $l$:
\begin{equation}\label{eq:noisy_expectation}
  f(\mathcal{C}(\bm{\theta}),\lambda)=f(\mathcal{C}(\bm{\theta}),0) + \sum_{k=1}^{l-1} c_k\lambda^k +\order{\lambda^{l}},
\end{equation}
where $c_k$ are constants dependent on the noise model and the quantum circuit.
The goal of quantum error mitigation is to recover the noiseless expectation value of the quantum circuit $f(\mathcal{C}(\bm{\theta}),0)=\langle O \rangle$, by post-processing the results of noisy measurements to minimize error effects without relying on additional qubits as in quantum error correction.

\textit{Zero-noise extrapolation}~(ZNE) is a widely used QEM technique to achieve this goal~\cite{temme2017error,li2017efficient}. The method involves artificially amplifying the noise levels and then extrapolating the expectation value to the zero-noise limit.
Suppose base noise rate~(true noise level of quantum device) is $\lambda$, the first step in ZNE is to select a set of noise levels, denoted as $\Lambda=\{\lambda_1\dots ,\lambda_l \mid \lambda_1=\lambda, \lambda_i<\lambda_{i+1}\}$, which can be realized in quantum devices.
In the original approach~\cite{temme2017error}, the various noise levels are achieved by adjusting the duration of gate operations. 
However, controlling gate time in practice is challenging, so alternative methods, such as hardware-agnostic implementations using identity insertions~\cite{dumitrescu2018cloud,he2020zero} or probability sampling~\cite{kim2023evidence}, are often employed.

The next step is to determine the coefficients $\{x_i\}$, that satisfy the following linear equations:
\begin{equation}\label{eq:linear_equation}
  \sum_{i=1}^{l}x_i=1, \quad \sum_{i=1}^{l}x_i\lambda_i^k=0 \quad \forall k\in\{1,\cdots,l-1\},
\end{equation}
which, by linear algebra, ensures a unique solution.
A weighted sum of $f(\mathcal{C}(\bm{\theta}),\lambda_i)$ with weights $\{x_i\}$, provides an estimate of the noiseless expectation value:
\begin{equation}\label{eq:estimator_noiseless}
  \begin{aligned}
    \widehat{\langle O \rangle}&=\sum_{i=0}^l x_i f(\mathcal{C}(\bm{\theta}),\lambda_i)\\
    &=\sum_{i=0}^l x_i \left(f(\mathcal{C}(\bm{\theta}),0) + \sum_{k=1}^{l-1} c_k\lambda_i^k+\order{\lambda_i^l}\right)\\
    &=f(\mathcal{C}(\bm{\theta}),0) + \order{\lambda_l^{l}}.
  \end{aligned}
\end{equation}
This approach is known as Richardson extrapolation~\cite{richardson1927viii,temme2017error}.
As noted in Ref.~\cite{lowe2021unified,he2020zero}, Richardson extrapolation is equivalent to a polynomial interpolation on the noisy expectation values. 
To illustrate this, consider $(\lambda_i,f(\mathcal{C}(\bm{\theta}),\lambda_i))_{i=1,\cdots,l}$ as discrete points to be interpolated.
A polynomial interpolation of order $(l-1)$ aims to find a polynomial $P(\lambda)=\sum_{k=0}^{l-1} c_k\lambda^k$ such that $P(\lambda_i)=f(\mathcal{C}(\bm{\theta}),\lambda_i)$ for all $i=1,\cdots,l$.
Using the Lagrange interpolation formula, the zero-order term is given by
\begin{equation}\label{eq:lagrange_interpolation}
  P(0)=c_0=\sum_{i=1}^{l} f(\mathcal{C}(\bm{\theta}),\lambda_i) \prod_{j\neq i} \frac{-\lambda_j}{\lambda_i-\lambda_j},
\end{equation}
where the coefficients $\{\prod_{j\neq i} \frac{-\lambda_j}{\lambda_i-\lambda_j}\}$ are precisely the solutions $\{x_i\}$ of the linear equation Eq.~\eqref{eq:linear_equation}. 
Therefore, the zero-order term $c_0$ in the polynomial interpolation matches the Richardson extrapolation result.

The estimator $\widehat{\langle O \rangle}$ in Eq.~\eqref{eq:estimator_noiseless} is fully determined by the noisy expectation values $\{f(\mathcal{C}(\bm{\theta}),\lambda_i)\}_{i=1,\cdots,l}$, as well as the artificially amplified noise levels $\{\lambda_i\}_{i=1,\cdots,l}$. 
Notably, when significant noise is present in the quantum device, the $(l-1)$-order approximate in Eq.~\eqref{eq:noisy_expectation} may become inaccurate. This results in the estimator $\widehat{\langle O \rangle}$ from Eq.~\eqref{eq:estimator_noiseless} potentially being a poor approximation of the behavior near $\lambda=0$. Consequently, several learning-based methods have been proposed to improve the estimation accuracy~\cite{czarnik2021error,czarnik2022improving,lowe2021unified}.

\textit{Variable-noise Clifford data regression}~(vnCDR) method~\cite{lowe2021unified}, similar to Eq.~\eqref{eq:estimator_noiseless}, derives the noiseless estimator through a linear transformation of the noisy expectation value:
\begin{equation}\label{eq:estimator_clifford}
  \widehat{\langle O \rangle}=\sum_{i=0}^l c_i f(\mathcal{C}(\bm{\theta}),\lambda_i),
\end{equation}
where $\{\lambda_i\}$ are artificially amplifying noise levels, and $\{c_i\}$ are coefficients learned from a training set.
The training set is assembled by a set of circuits $\{\mathcal{C}_1,\cdots \mathcal{C}_K\}$, that resemble the target circuit $\mathcal{C}(\bm{\theta})$, but only encompass a constant number of non-Clifford gates.
Subsequent to this process, the corresponding noisy expectation values $\{f(\mathcal{C}_j,\lambda_i)\}$ are measured. 
The extended Gottesman-Knill theorem~\cite{gottesman1998heisenberg,bravyi2016improved,howard2017application,aaronson2004improved} ensures that the expectation values of quantum circuits dominated by Clifford gates can be classical simulated efficiently, thus allowing the noiseless expectation values $\{f(\mathcal{C}_j,0)\}$ to be obtained.
The coefficients $\{c_i\}$ are learned from the Clifford training set $\{(f(\mathcal{C}_j,\lambda),f(\mathcal{C}_j,0))\}$ by solving the least-square problem:
\begin{equation}\label{eq:least_square}
  \mathrm{arg}\operatorname*{min}_{c_1,\cdots,c_l} \sum_{j=1}^{K} \left[\sum_{i=1}^{l} c_i f(\mathcal{C}_j,\lambda_i)-f(\mathcal{C}_j,0)\right]^2.
\end{equation}
This approach directly learns the linear transformation from noisy to noiseless expectation values. Empirical findings demonstrate that the vnCDR method can provide more scalable corrections than Richardson extrapolation in certain scenarios~\cite{czarnik2021error}.

\textit{Learning-based PEC} method~\cite{strikis2021learning} is another learning-based approach.
The first step involves selecting a set of operations for inserting Pauli gates into circuits, denoted as $\{g_1,\cdots,g_W\}$, for the sake of convenience, assigning $g_1$ to signify no insertion.
By performing these insertions, the target circuit $\mathcal{C}(\bm{\theta})$ is modified, yielding a series of circuits given by $\{g_1(\mathcal{C}(\bm{\theta})),\cdots,g_W(\mathcal{C}(\bm{\theta}))\}$.
This method eliminates the need for multiple noise levels, requiring only the measurement of the expectation values of these circuits at the base noise level $\lambda$ on quantum device, represented by $\{ f(g_w(\mathcal{C}(\bm{\theta})), \lambda) \}_{w=1, \dots, W}$.
The noiseless expectation value is then estimated as a linear combination of the noisy values, given by:
\begin{equation}
\widehat{\langle O \rangle}(\bm{\theta})=\sum_{w=1}^{W} c_w f(g_w(\mathcal{C}(\bm{\theta})),\lambda),
\end{equation}
where $\{c_w\}$ are the coefficients determined from the training set.
The training set comprises Clifford circuits $\{\mathcal{C}_1,\cdots \mathcal{C}_K\}$.
The coefficients are learned from the training set by solving:
\begin{equation}\label{eq:least_square_low}
 \mathrm{arg}\operatorname*{min}_{c_1,\cdots,c_w} \sum_{j=1}^{K} \left[\sum_{w=1}^{W} c_w f(g_w(\mathcal{C}_j),\lambda)-f(\mathcal{C}_j,0)\right]^2.
\end{equation}

In the following, we extend the training set beyond the Clifford circuits by introducing Clifford perturbation data.

\section{Main framework}
\label{sec:main}

In this section, we propose a learning-based error mitigation framework called the \textit{Clifford Perturbation data regression}. 
The main idea is to construct training sets using Clifford perturbation circuits, which are close to Clifford circuits and denoted as PQCs $\mathcal{C}(\bm{\theta})$, where the rotation angles are constrained to a small range $\Theta = \{\bm{\theta} \mid |\theta_i| \leq \theta_*, i = 1, \dots, L \}$, with $\theta_*$ being a small constant.

\subsection{Clifford Perturbation Approximation}\label{sec:sim}

We first introduce an efficient method to approximate the expectation value of a Clifford perturbation circuit, known as the \textit{Sparse Pauli Dynamic}~(SPD) method~\cite{beguvsic2023simulating}. We would also investigate the errors associated with truncation.

When the rotation angle $\theta$ take values from $\{\frac{k\pi}{4}\mid k\in \mathbb{Z}\}$, the Pauli rotation gate $e^{-i\frac{\theta}{2} P}=\cos{\frac{\theta}{2}}\mathbb{I}+i\sin{\frac{\theta}{2}}P$ belongs to the Clifford group. 
Thus, we assume that the rotation angle $\theta_i$ in Eq.~\eqref{eq:parameterized_circuit} lies within the range $[-\frac{\pi}{4},\frac{\pi}{4}]$. 
For any $\theta_i$ outside this range, there exists $\theta'_i\in [-\frac{\pi}{4},\frac{\pi}{4}]$ such that $\theta'_i+\frac{k\pi}{2}=\theta_i$. 
In this case, the unitary operator $U_i(\theta_i)$ can be written as $U_i(\theta_i)=e^{-i\frac{\theta}{2} P_i}C_i=e^{-i\frac{\theta'_i}{2} P_i}e^{-i\frac{k\pi}{4} P_i}C_i$, where $e^{-i\frac{k\pi}{4} P_i}C_i$ is a Clifford operator. By substituting $C_i$ with $C_i'=e^{-i\frac{k\pi}{4} P_i}C_i$, the unitary operator becomes $U_i(\theta_i)=e^{-i\frac{\theta'_i}{2} P_i}C_i'$, where $\theta'_i\in [-\frac{\pi}{4},\frac{\pi}{4}]$.

Given the restriction of rotation angles to the narrow range $[-\frac{\pi}{4},\frac{\pi}{4}]$, the Pauli rotation can be interpreted as a perturbation of Clifford operators. The expectation value could be efficiently approximated by truncating higher-order perturbation terms.
Specifically, in the Heisenberg picture, the expectation value is reformulated as $\langle O \rangle = \tr{\rho \tilde{O}}$, with $\tilde{O} = {U}_1(\theta_1)^{\dagger} \cdots  {U}_L(\theta_L)^{\dagger} O {U}_L(\theta_L)  \cdots {U}_1(\theta_1)$ representing the Heisenberg-evolved observable. 
The Heisenberg evolution of a Pauli rotation gate $e^{-i\frac{\theta}{2} P}$ acting on a Pauli operator $\sigma$ is described by:
\begin{equation}
  e^{i \frac{\theta}{2} P} \sigma e^{-i \frac{\theta}{2} P} = 
  \begin{cases}
  \sigma, & [\sigma,P] = 0, \\
  \cos(\theta) \sigma + i \sin(\theta) P\sigma & \{\sigma, P\} = 0.
  \end{cases} 
\end{equation}
In the case of a Pauli observable $O$, the first-step Heisenberg evolution ${U}_L(\theta_L)^{\dagger} O {U}_L(\theta_L)=C^\dagger_Le^{-i\frac{\theta}{2} P_L}Oe^{-i\frac{\theta}{2} P_L}C_L$ simplifies to $C_L^\dagger O C_L$, when $[O,P_L]=0$. Conversely, for $\{O,P_L\} = 0$, it becomes $\cos(\theta_L)C_L^\dagger O C_L + i\sin(\theta_L) C_L^\dagger P_L O C_L$.
By iterating this process over all unitary operators $U_i(\theta_i)$, the resulting Heisenberg-evolved observable $\tilde{O}$ takes the form of a linear combination of Pauli operators, denoted as $\tilde{O} = \sum_\sigma c_\sigma \sigma$, where $c_\sigma$ is a multivariate trigonometric monomial involving products of $\cos(\theta_i)$ and $\sin(\theta_i)$.
To truncate the perturbation terms to order $M$, we discard the terms $c_\sigma \sigma$ if the number of $\sin$ factors in $c_\sigma$ exceeds $M$. 
The truncated observable is denoted by $\tilde{O}_M=\sum_{\{\sigma\mid \abs{c_\sigma}_{\sin}\leq M\}} c_\sigma \sigma$, where $\abs{c_\sigma}_{\sin}$ denotes the number of $\sin$ factor in $c_\sigma$.
The approximate expectation value is then given by:
\begin{equation}\label{eq:sim_noiseless}
  \langle O \rangle ^{(M)}= \tr{\rho \tilde{O}_M}.
\end{equation}

For a truncation number $M$, if the number of non-zero elements in $\rho=\sum_{a,b}\rho_{a,b}\ketbra{a}{b}$, as well as the number of Pauli operators $\{\sigma\}$ that linearly compose $O$, are both polynomially related to the number of qubits $n$, denoted as $\mathrm{Poly}(n)$, then the computational cost of the algorithm is bounded by $\order{\mathrm{Poly}(n) L^{M+1}}$.

The algorithm is summarized as follows:
\begin{algorithm}[H]
  \caption{SPD: Estimate Expectation Value by Truncating Clifford Perturbation}\label{ALGORITHM_HS}
  \begin{algorithmic}
    \State Set $\langle O \rangle ^{(M)}=0$.
    \State Enumerate $\sigma_L$ as all Pauli operator with non-zero coefficient $c_L$ in $O$.
    \State Set sin counter $\#\sin_L=0$.
    \For{candidates of $\sigma_L$ in $O$}
    \State{Set candidate list $Cl_L=\{\}$}
    \If{$[\sigma_L,P_L]=0$} 
      \State Add $(C_L^\dagger \sigma_L C_L,c_L,\#\sin_L)$ to $Cl_L$.
    \Else
      \State Add $(C_L^\dagger \sigma_L C_L,c_L\cos(\theta_L),\#\sin_L)$ to $Cl_L$.
      \State Add $(iC_L^\dagger P_L \sigma_L C_L,c_L\sin(\theta_L),\#\sin_L+1)$ to $Cl_L$.
  \EndIf 
      \For{candidates of $(\sigma_{L-1},c_{L-1},\#\sin_{L-1})$ in $Cl_L$}
      \State{Eliminate cases with $\#\sin_{L-1} > M$.}
      \State{Set candidate list $Cl_{L-1}=\{\}$}
        \State{\vdots}
        \For{candidates of $(\sigma_{0},c_{0},\#\sin_{0})$ in $Cl_1$}
        \State{Eliminate cases with $\#\sin_{0} > M$.}
        \State Update $\langle O \rangle ^{(M)}=\langle O \rangle ^{(M)}+c_0\tr{\rho\sigma_0}$.
        \EndFor
    \EndFor
    \EndFor
    \State Output the approximate expectation value $\langle O \rangle ^{(M)}$.
  \end{algorithmic}
\end{algorithm}

Additionally, when the rotation parameters are located in the small angle space $\Theta=\{\bm{\theta}\mid \abs{\theta_i}\leq \theta_* ,i=1,\cdots,L \}$, the truncation error can be upper bounded, as summarized in the following theorem. More details are provided in the Suppl. Mat.~\ref{ap:sec:CPA}.

\begin{theorem}\label{thm:truncation_error}
  For any given $\delta>0$, $M>0$ satisfies $\frac{\ln{1+\frac{\delta}{2}}}{L-M}\leq \frac{\ln{2}}{M}$, and Pauli observable $O$, if $\theta_*=\frac{\ln{1+\frac{\delta}{2}}}{L-M}=\order{\frac{1}{L-M}}$, then the truncation error satisfies $\abs{\langle O \rangle - \langle O \rangle ^{(M)}}\leq \delta$ for all $\bm{\theta}\in \Theta$.
  In average case, if $\theta_*=\sqrt{\frac{3\ln{1+\frac{\delta}{2}}}{L-M}}=\order{\frac{1}{\sqrt{L-M}}}$, then the mean square error $\mathbb{E}_{\bm{\theta}\in \Theta}[(\langle O \rangle - \langle O \rangle ^{(M)})^2]\leq \delta$.
\end{theorem}

\begin{figure*}[htb!]
 \centering%
 \includegraphics[width = 1.8\columnwidth]{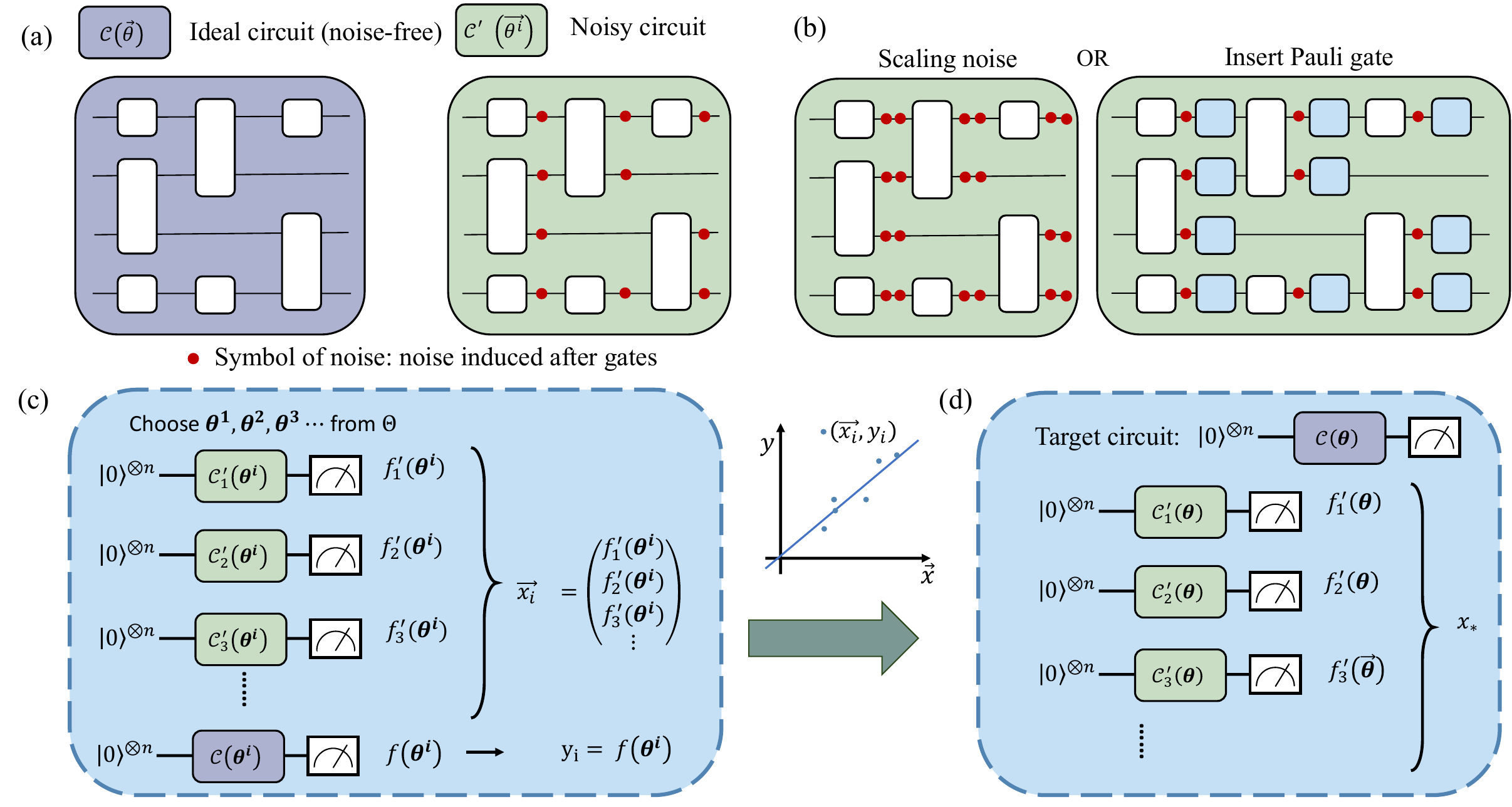}
 \caption{\justifying
 The framework of our Clifford Perturbation data regression~(CPDR). 
 (a) Ideal circuit model and noise circuit model. 
 (b) Noisy circuits used to estimate the noise-free expectation value: In CPDR-ZNE, the noisy circuits are generated by artificially amplifying the noise level, while in CPDR-PEC, noisy circuits are generated by inserting Pauli gates into the target circuit.
 (c) The training set construction process in the CPDR framework: The circuits in the training set share the same structure as the target circuit, with rotation angles selected from a constrained space $\Theta$. The noiseless reference expectation values are then obtained using the SPD method.
 (d) The mapping between the noisy expectation value and the noise-free expectation value is learned by fitting the training set, and this map is subsequently applied to the target circuit to mitigate the noise.}%
 \label{fig:whole_picture_MPO_QEM}
\end{figure*}

\subsection{Clifford Perturbation data regression framework}

In this subsection, we explain the CPDR framework as shown in Fig.~\ref{fig:whole_picture_MPO_QEM}. 
The training set construction is a crucial step for learning-based error mitigation protocols~\cite{czarnik2021error,strikis2021learning,lowe2021unified,czarnik2022improving}, which requires a collection of classically simulatable circuits with noiseless expectation values.
In previous works, Clifford circuits are commonly used in vnCDR and learning-based PEC protocol.
However, their sparsity in the unitary space limits their ability to approximate target circuits accurately.
To address this, we enhance the training set by incorporating Clifford perturbation circuits.

Specifically, for any quantum circuit $\mathcal{C}(\bm{\theta})$, the rotation angles $\bm{\theta}=(\theta_1,\cdots,\theta_L)$ can be assumed to be in range $[-\frac{\pi}{4},\frac{\pi}{4}]$.
We denote the noiseless expectation value as $f(\mathcal{C}(\bm{\theta}),0)$, and the noisy expectation value as $f(\mathcal{C}(\bm{\theta}),\lambda)$, where $\lambda$ represents the base noise rate of quantum device. 
The circuits in the training set share the same structure as the target circuit $\mathcal{C}(\bm{\theta})$, but their rotation angles are constrained to a small range. The training set can be expressed as:
\begin{equation}
  \left\{\mathcal{C}(\bm{\theta}^k)\mid \abs{\theta^k_i}\leq \theta_*, i=1,\cdots,L\right\}_{k=1,\cdots,K},
\end{equation}
where $\theta_*$ is a small constant and $K$ is the number of circuits in training set.
As shown in Theorem \ref{thm:truncation_error}, when $\theta_*\leq \frac{\ln{1+\frac{\delta}{2}}}{L-M}$ the truncation error satisfies $\abs{f(\mathcal{C}(\bm{\theta}^k),0) - \langle O \rangle ^{(M)}(\bm{\theta}^k)}\leq \delta$, where $\langle O \rangle ^{(M)}(\bm{\theta}^k)$ represents the truncated approximate expectation for the Clifford perturbation circuit $ \mathcal{C}(\bm{\theta}^k) $, as defined in Eq.~\eqref{eq:sim_noiseless}.
Therefore, when taking appropriate $\theta_*$, we can employ the SPD method to efficiently approximate the noiseless expectation value of the circuits in training set with a controllable truncation error.

By selecting a set of noise levels $\Lambda=\{\lambda_1,\cdots,\lambda_l\}$, a candidate of noiseless estimator is given by $\widehat{\langle O \rangle}(\bm{\theta})=\sum_{i=1}^{l} c_i f(\mathcal{C}(\bm{\theta}),\lambda_i)$, where $\{c_i\}$ are the coefficients to be determined from the training set by solving the least-square problem:
\begin{equation}\label{eq:L2_loss_clifford_perturbation}
  \begin{aligned}
  \mathrm{arg}\operatorname*{min}_{c_1,\cdots,c_l} &\sum_{k=1}^{K} \left(\sum_{i=1}^{l} c_i f(\mathcal{C}(\bm{\theta}^k),\lambda_i)-\langle O \rangle^{(M)}(\bm{\theta}^k)\right)^2\\ &+ \alpha \sum_{i = 1}^l c_i^2,
  \end{aligned}
  \end{equation}
We refer to this protocol as the Clifford Perturbation data regression Zero Noise Extrapolation~(CPDR-ZNE).
To achieve better generalization and reduce over-fitting, we use Ridge Regression~\cite{horel1962application,hoerl1970ridge,mcdonald2009ridge}, which introduces an $L^2$ regularization term to the loss function.
Where $\alpha$ is a hyperparameter of the regularization term.

We also enhanced the learning-based PEC method by incorporating the Clifford perturbation circuits into the training set, naming it Clifford Perturbation data regression probabilistic error cancellation~(CPDR-PEC).

By selecting a set of operations to insert Pauli gates $\{g_1,\cdots,g_W\}$, where $g_1$ represents no additional gate insertion, modified circuits are generated as $\{g_1(\mathcal{C}(\bm{\theta})),\cdots,g_W(\mathcal{C}(\bm{\theta}))\}$.
The noiseless estimator is given by a linear transformation of the noisy expectation values of these circuits $\widehat{\langle O \rangle}(\bm{\theta})=\sum_{w=1}^{W} c_w f(g_w(\mathcal{C}(\bm{\theta})),\lambda)$, where $\{c_w\}$ are the coefficients determined solving the following optimization problem:
\begin{equation}\label{eq:least_square_clifford_perturbation_PEC}
 \mathrm{arg}\operatorname*{min}_{c_1,\cdots,c_l} \sum_{k=1}^{K} \left(\sum_{w=1}^{W} c_w g_w(\mathcal{C}(\bm{\theta}^k))-\langle O \rangle^{(M)}(\bm{\theta}^k)\right)^2.
\end{equation}

Our protocol differs from previous learning-based methods by using Clifford perturbation circuits instead of Clifford circuits for the training set.
Since both Clifford and near-Clifford circuits are efficiently simulatable via SPD, our approach inherently extends prior methods. Near-Clifford circuits offer better generalization capabilities, thereby the accuracy is improved with this modification, as shown in Sec.~\ref{sec:Numerical}.

\section{Numerical benchmarks}
\label{sec:Numerical}

In this section, we first benchmark the SPD simulator on hardware-efficient circuits proposed in Ref.~\cite{sim2019expressibility}. 
We validate Theorem \ref{thm:truncation_error} and demonstrate that, numerically, the SPD approach is able to maintain high precision over a broader range of rotation angles in PQCs.
Next, we numerically compare CPDR-ZNE and CPDR-PEC methods with original ZNE and learning-based PEC, using Trotterized time evolution circuits.  The results clearly indicate the advantages of our approach.
Finally, we apply CPDR-ZNE method to the experimental data from IBM's 127 qubits Eagle processor, as reported in Ref.~\cite{kim2023}, illustrating the usability and effectiveness of our protocol on large-scale quantum hardware platforms.

\subsection{Benchmarks of SPD}

\begin{figure}[!ht]
    \centering
\includegraphics[width=0.25\textwidth]{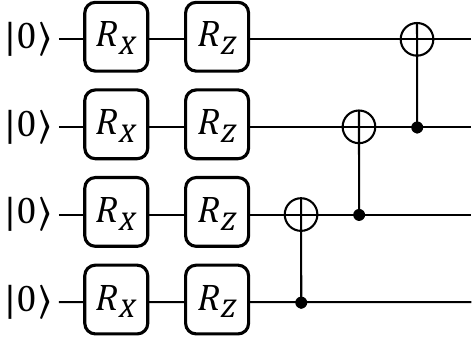}
    \caption{\justifying
    The architecture of the block in the hardware-efficient circuit.}
\label{fig:hardware_block}
\end{figure}

\begin{figure}[!ht]
    \centering
\includegraphics[width=0.48\textwidth]{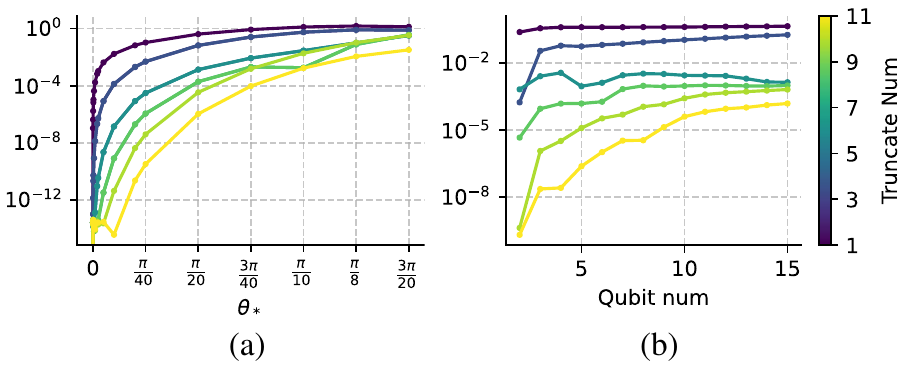}
    \caption{\justifying 
    Benchmarks for the SPD simulator with hardware efficient circuits, which consist $5$ blocks shown in Fig.~\ref{fig:hardware_block}. The error is defined as the difference $\abs{\langle O \rangle - \langle O \rangle^{(M)}}$, where the rotation angle of each gate within each block is uniform and denoted as $\theta_*$.
    (a) shows the error for circuits with a fixed qubit number $n=15$, plotted against the rotation angle $\theta_*$.
    (b) illustrates the error for a given rotation angle $\theta_*=\pi/20$, with respect to the number of qubits $n$.}
\label{fig:bench_mini_AQC}
\end{figure}

The hardware-efficient circuits consist of multiple repeated blocks, with the architecture of each block shown in Fig.~\ref{fig:hardware_block}.
In our setting, we configure the circuits to contain five blocks, and the number of qubits ranges from 2 to 15. Additionally, the rotation angles for each gate are identical throughout the circuits and are located in the range $[0,\pi/4]$.

In Theorem~\ref{thm:truncation_error}, we demonstrated that if the maximum angle is restricted to $\theta_*=\frac{\ln{1+\frac{\delta}{2}}}{L-M}$, then the truncated error can be bounded by $\abs{\langle O \rangle - \langle O \rangle ^{(M)}}\leq \delta$ for angles within the small-angle space $\bm{\theta}\in\Theta=\{\bm{\theta}\mid \abs{\theta_i}\leq \theta_* ,i=1,\cdots,L \}$.
For attaining an accuracy of $10^{-2}$, the analytical result suggests that $\theta_*$ should be approximately $\ln{1+\frac{10^{-2}}{2}}\approx 0.005$, indicating a very narrow interval.
This presents a difficult in preparing sufficient training data within $\Theta$ for the CPDR framework.
However, when the truncation number $M$ is moderately large (e.g $M=5,7$), the numerical results remain highly accurate over a much broader range of angles, rendering this approach suitable for practical scenarios.

As depicted in Fig.~\ref{fig:bench_mini_AQC}(a), when $\theta_*\leq \pi/20$, the truncated error decreases rapidly as the truncation number $M$ increases. 
For an accuracy of $10^{-2}$, $\theta_*\leq \pi/20$ is sufficient for $M\geq 5$.
Furthermore , when $M=11$, the truncated error remains within $10^{-2}$ until $\theta_*=\pi/8$. 
This indicates that slightly larger values of $M$ can considerably enhance the accuracy beyond the theoretical bound in practical scenarios.

To assess scalability for large-scale quantum circuits, we evaluate the truncation error $\abs{\langle O \rangle - \langle O \rangle ^{(M)}}$ for varying truncation number $M$ and qubit numbers ranging in $[2,15]$.
Considering the SPD simulator is employed to provide training data for the CPDR framework, we set the angle of each rotation gate to an appropriate value for training data preparation.
Based on empirical results from numerical tests, we recommend setting $\theta_* = \pi/20$ and obtaining training data from $\Theta$ to balance the training set size and simulation accuracy.
As shown in Fig.~\ref{fig:bench_mini_AQC}(b), as the number of qubits increases, the error keeps below $10^{-2}$ across $M\geq 5$, indicating that the SPD simulator is scalable for large-scale quantum circuits.

\subsection{Numerical results of CPDR framework}

\begin{figure*}[!ht]
    \centering
    \includegraphics[width=0.98\textwidth]{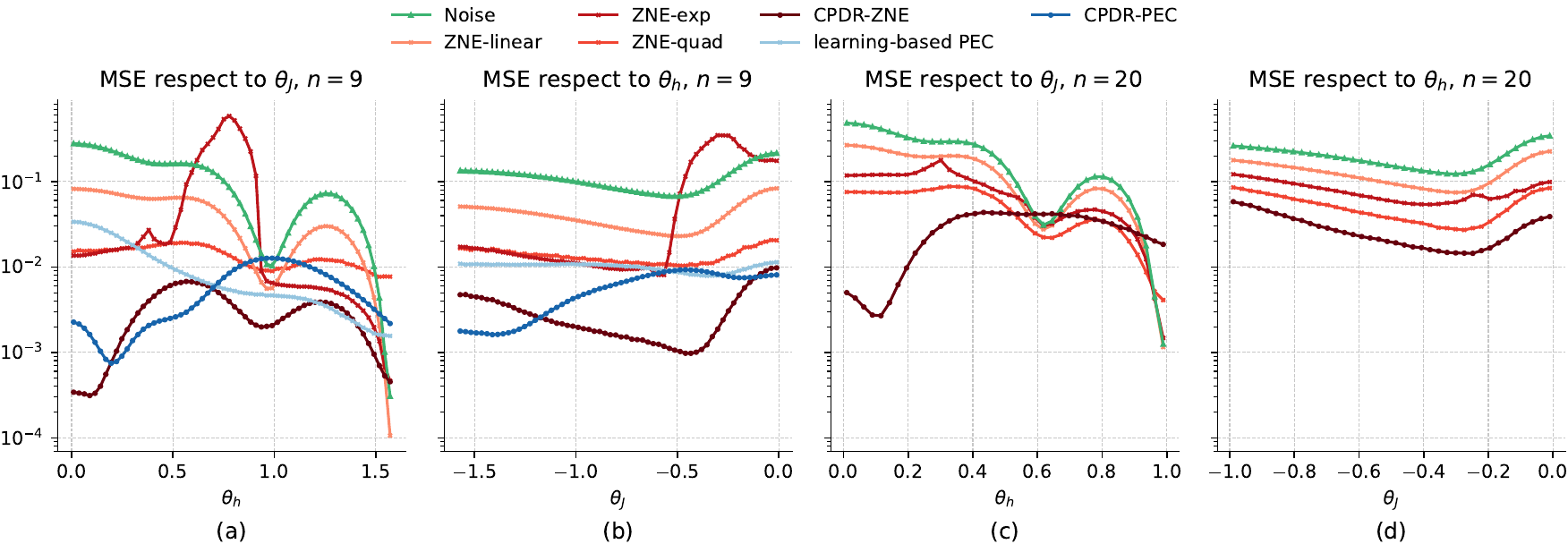}
    \caption{\justifying
    Comparisons are made between ZNE~(with linear, quadratic and exponential fitting function, denoted as ``ZNE-linear'', ``ZNE-quad'' and ``ZNE-exp''), learning-based PEC, CPDR-ZNE and CPDR-PEC on the time evolution circuits of the Ising model. 
    The expectation values from the noise circuit are also presented for comparison, labeled as ``noise''.
    For circuit with $9$ qubits and $5$ Trotter steps, (a) displays the mean squared error (MSE) over $\theta_J$, defined as $\mathbb{E}_{\theta_J}\abs{\langle O \rangle-\widehat{\langle O \rangle}}^2$, for various $\theta_h$. (b) presents the MSE over $\theta_h$ for different $\theta_J$.
    The corresponding results for circuit with $20$ qubits $8$ Trotter steps are shown in (c) and (d).
    }
    \label{fig:estimate_MSE_Ising_10_20}
\end{figure*}

In this subsection, we benchmark the CPDR-ZNE and CDPR-PEC methods with the original ZNE and learning-based PEC protocols by Trotterized time evolution circuits of the Ising model.

The Hamiltonian of the Ising model is given by
\begin{equation}\label{target_H}
  H = - \sum_{\langle i, j
\rangle \in E}J_{ij} Z_i Z_j +  \sum_i h_iX_i,
\end{equation}
where $i < j$, $ \langle \, , \, \rangle$ represents nearest-neighbour spin pair, $J_{ij}> 0$ is the coupling constant, and $h_i$ denotes the global transverse field. The set \( E \) represents the edges of the lattice.

The time dynamics of the Ising model is governed by the Schrödinger equation:
\begin{equation}
  \begin{aligned}
  \frac{\partial}{\partial t} \ket{\psi(t)} = -i H \ket{\psi(t)}.
  \end{aligned}
\end{equation}
where $\ket{\psi(t)}$ represents the state of the system at time $t$, the initial state is set to $\ket{\psi(0)} = \ket{0}^{\otimes n}$.
By discretizing the evolution time $T$ into $N$  steps, the time evolution can be simulated in quantum circuits using the first-order Trotter decomposition of the time-evolution operator:
\begin{equation}\label{eq:time_evolution}
\small
  \begin{aligned}
    \ket{\psi(T)}&=\prod^N_{k=1} \exp{-i \frac{T}{N} \cdot H} \ket{\psi(0)}\\
    &\approx\prod^N_{k=1}\left(\prod_{\langle i, j\rangle \in E} \mathrm{e}^{i \frac{J_{ij}T}{N} Z_i Z_j}\prod_i\mathrm{e}^{-i \frac{h_iT}{N} X_i }\right) \ket{\psi(0)}\\
  &=\prod^N_{k=1}\left(\prod_{\langle i, j\rangle \in E} \mathrm{R}_{Z_i Z_j}(-\frac{2 J_{ij}T}{N})\prod_i \mathrm{R}_{X_i}(\frac{2Th_i}{N}) \right)\ket{\psi(0)}, \\
  \end{aligned}
\end{equation}
where $\mathrm{R}_{Z_i Z_j}$ and $\mathrm{R}_{X_i}$ denotes the $Z Z$ and $X$ rotation gates, respectively. 
For simplicity, we consider a one-dimensional Ising model on a 1D chain, where $E = \{\langle i,i+1\rangle: i = 1,2,\cdots, n-1\}$. 
We assume that $h_i = h$ for all qubits $i$, and $J_{i,j} = J$ for each pair $\langle i, j\rangle\in E$, with $ J > 0 $ and $ h > 0 $.
Consequently, the rotation angles for all $\mathrm{R}_{X}$ gates are identical and denoted by $\theta_h = \frac{2Th}{N}$, while the rotation angles for all $\mathrm{R}_{ZZ}$ gates are also identical, denoted by $\theta_J = -\frac{2JT}{N}$. 
In our simulations, the CPDR training set consists of circuits with $\theta_h \in [0, \pi/20] \cup [9\pi/20, \pi/2] $ and $ \theta_J \in [-\pi/20, 0] \cup [-\pi/2, -9\pi/20] $.
The observable is defined as $O = \sum_{i = 1}^n Z_i /n$, and the expectation of this observable, denoted $ M_z $, represents the global magnetization along the $ \hat{z} $-axis.

We adopt a gate-based noise model that closely resembles the noise characteristics of realistic quantum devices.
For each gate, we introduce the depolarizing noise following the gate operation and account for thermal relaxation during the gate's duration. We assume the coherence times of each qubit are $T_1 = 100 \, \mu s$, $T_2 = 50 \, \mu s$. For single-qubit gates, the depolarizing noise intensity is $\lambda_{\text{single}}=0.01$ and the gate operation time is $t_{\text{gate}} = 300 \, ns$. For two-qubit gates, the depolarizing noise intensity is $\lambda_{\text{double}}=0.04$ and the gate operation time is $t_{\text{gate}} = 800 \, ns$. 
We assume that the noise in this model represents the base noise level, characterized by constant value $ \lambda$.

We apply both CPDR-ZNE and CPDR-PEC protocol to the Ising model time evolution circuits under the noise model described above.
In comparison, we also apply ZNE and learning-based PEC to the same noisy circuits.
The noise levels applied in the CPDR-ZNE and ZNE protocol are scaled using factors $ G = 1, 1.2, 1.6$ as outlined in the mitigation experiment~\cite{kim2023evidence,PhysRevA.106.062436}, so that $ \Lambda = \{\lambda, 1.2\lambda, 1.6\lambda\}$. 
For both CPDR-PEC and learning-based PEC, we selected $20$ configurations for single Pauli-X or Pauli-Z gate inserted.
The training set for the learning-based PEC includes 2048 Clifford circuits. 
In all cases, the shot count is set to $ 10^4 $. More details are provided in the Suppl. Mat.~\ref{sec:CPDR}.

We first consider a circuit with qubit $n=9$ and Trotter step $N = 5$, and the results are shown in Fig.~\ref{fig:estimate_MSE_Ising_10_20}(a) and (b).
The methods in CDPR-framework exhibit a lower MSE, demonstrating an obvious advantage over other QEM protocols.
Subsequently, another benchmark is performed on a circuit with qubit $n=20$ and Trotter step $N=8$.
Given the relatively steep resource costs associated with learning-based PEC and CPDR-PEC, we exclude these from the comparison for this instance.
The outcomes are shown in Fig.~\ref{fig:estimate_MSE_Ising_10_20}(c) and (d).
Despite the increased circuit size, CPDR-ZNE still exhibits a clear advantage, suggesting that our protocol is scalable.

Furthermore, we benchmark circuits with 8 and 10 qubits under varying base noise rates. Both CPDR-ZEN and CPDR-PEC demonstrate superior accuracy across these tests. 
Additionally, since the CPDR protocols rely on noiseless expectations obtained via SPD for constructing the training set, potential inaccuracies in SPD could introduce errors. 
To investigate this, we examine the impact of simulation errors on the accuracy of the CPDR protocols and found that the CPDR protocols are robust against such errors, with no significant reduction in accuracy.
More details please see Suppl. Mat.~\ref{sec:ap:Numerical}.

\subsection{Experimental results on IBM's Eagle processor}

\begin{figure*}[!ht]
    \centering
    \includegraphics[width=0.98\textwidth]{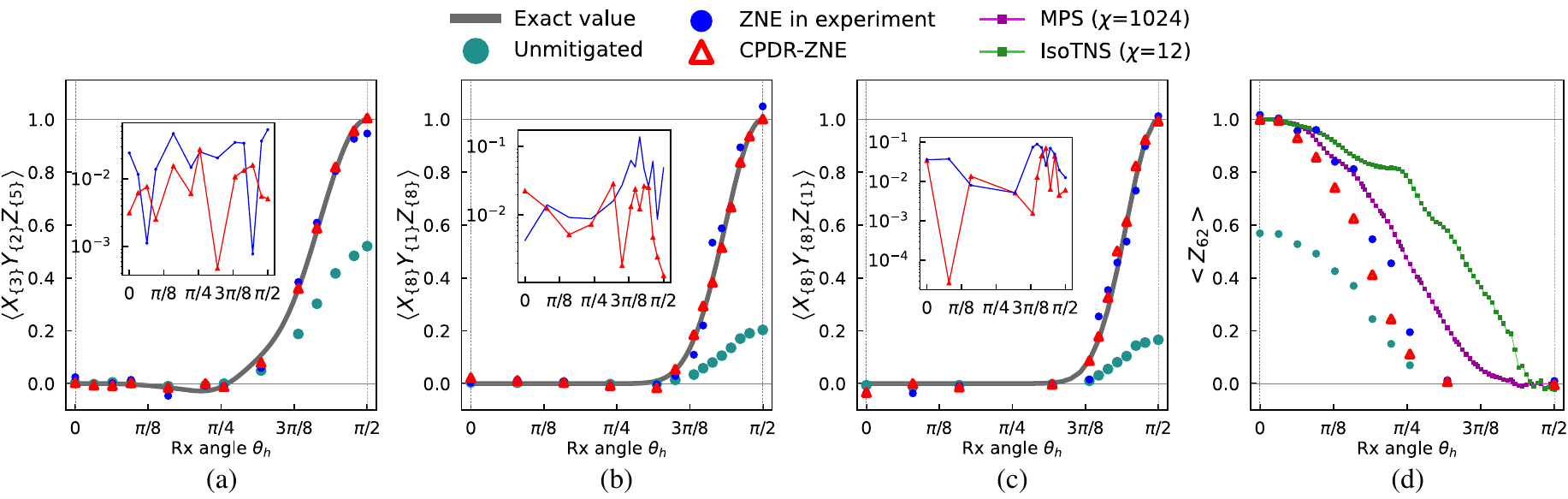}
    \caption{\justifying
    The comparison between the CPDR-ZNE and the mitigation protocol employed in Ref.~\cite{kim2023evidence}(labeled as ``ZNE in experiment''), which automatically selects between exponential, linear fits and noise expectation based on measurement results. 
    The raw experimental data from Ref.~\cite{kim2023} includes noisy expectation values obtained from IBM's Eagle processor, with varied noise scaling factors $G$, and we present the noise data with $G = 1$, labeled as ``Unmitigated''.
    In (a), the results for the circuit with $5$ Trotter steps and measurements on $X_{\{3\}} Y_{\{2\}} Z_{\{5\}}$ $\equiv X_{13,29,31} Y_{9,30} Z_{8,12,17,28,32}$ are shown.
    In (b), the results correspond to the circuit with $5$ Trotter steps and measurements on $X_{\{8\}} Y_{\{1\}} Z_{\{8\}}$ $\equiv X_{37,41,52,56,57,58,62,79} Y_{75} Z_{38,40,42,63,72,80,90,91}$.
    In (c), the results pertain to the circuit with $5$ Trotter steps and measurements on $X_{\{8\}} Y_{\{8\}} Z_{\{1\}}$ $\equiv X_{37,41,52,56,57,58,62,79} Y_{38,40,42,63,72,80,90,91} Z_{75}$, with an additional layer of $R_x\left(\theta_h\right)$ gates applied before measurement.
    In (d), the results are for the circuit with $20$ Trotter steps and measurements on $ Z_{62}$.
    For (a), (b) and (c), noiseless expectation values provided in Ref.~\cite{kim2023evidence} and Ref.~\cite{beguvsic2024fast} are included and labeled as ``Exact value''.
    The inset figures display the errors for both the IBM ZNE mitigation results (blue points) and the CPDR-ZNE results (red triangles).
    The x-axis in the inset figures correspond to $\theta_h$, and the y-axis represents the error. 
    For (d), since the exact value of the noiseless circuit is unavailable, we provide simulation results from isoTNS and MPS from Ref.~\cite{kim2023evidence} as alternatives.
    }
    \label{fig:estimate_AQC}
\end{figure*}

Ref.~\cite{kim2023evidence} demonstrates the experiment of implementing Trotterized time evolution circuits on the $127$ qubits IBM Eagle processor, for $\theta_h \in [0,\pi/2]$ and $\theta_J = -\pi/2$.
The interaction set $ E $ in the Hamiltonian, given in Eq.~\eqref{target_H}, is determined based on the topology of the IBM Eagle processor.
In this section, we examine four Trotterized time evolution circuits presented in Ref.~\cite{kim2023evidence}.
The circuit depths for these four configurations are $20,20,21,80$, and they are measured on Pauli observables with weights $10,17,17,1$, respectively.
For further details, please refer to Suppl. Mat.~\ref{sec:ap:Numerical:IBM}.

To estimate the noiseless expectation values, the experiment utilizes noise scaling factors \(G = 1, 1.3, 1.6\) for the third circuit and \(G = 1, 1.2, 1.6\) for the others, applying the ZNE which automatically selects between exponential, linear fits and noise expectation based on measurement results. 
We validate our CPDR-ZNE protocol by leveraging quantum device's experiment data as supplied in Ref.~\cite{kim2023}, which is the original dataset utilized in original mitigation protocol during the experiment.

For each circuit, we select the four parameter points, $\theta_h$, closest to the Clifford circuit to serve as the training set for CPDR-ZNE.
Specifically, these include the two points nearest to $0$ on the left and the two points nearest to $\frac{\pi}{2}$ on the right. 
The noiseless expectation references in training set are obtained by SPD.

Fig.~\ref{fig:estimate_AQC}~(a-c), which include the ideal noiseless expectations as references, demonstrate crucial improvements in accuracy when taking the CPDR-ZNE compared to the extrapolation mitigation used in the original experiments.
Fig.~\ref{fig:estimate_AQC}~(d) corresponds to a circuit with depth 80, the ideal expectations are not accessible. Consequently, only the results from the CPDR-ZNE and extrapolation mitigation methods are presented.
Compared to the extrapolation-based mitigation method, the CPDR-ZNE method consistently outperforms it across all three circuits in most cases. 
Importantly, the CPDR-ZNE protocol does not require additional physical resources to complete this task, demonstrating its ease of implementation. 
Furthermore, these results provide strong evidence of CPDR’s effectiveness in large-scale quantum hardware experiments.

\section{Conclusions and Discussions}
\label{sec:Discussion}

In this paper, we introduce the novel CPDR framework for learning-based error mitigation. 
This approach enhances the training set by incorporating Clifford perturbation circuits, which are parameterized quantum circuits with rotation gate angles constrained below a small threshold.

We utilize the SPD simulator to obtain ideal circuit expectation values in the training set. To ensure the simulation's accuracy, we provide a theoretical analysis, rigorously deriving the effective parameter range for reliable simulations.
Additionally, we conduct a numerical error analysis for the SPD with hardware-efficient circuits and demonstrate that the SPD method produces highly accurate computational results, even over a wide range of rotation angles.

In comparison to similar learning-based protocols, the CPDR framework admits that the parameters of the circuits in the training set are more closely aligned with those of the target circuit, resulting in more accurate error mitigation.
We develop the CPDR-ZNE and CPDR-PEC protocols based on the CPDR framework.
We highlight the advantages of our approach by numerically employing the CPDR protocols for the Trotter evolution circuits of the Ising model and comparing it with various QEM protocols. 
Finally, we apply the CPDR-ZNE protocol to the data from IBM’s 127-qubit Eagle processor and achieve better mitigation accuracy than the ZNE protocol employed in the original experiment~\cite{kim2023evidence}. This provides strong evidence of the effectiveness of our approach on large-scale quantum device.

In this work, we construct a training set with Clifford perturbation circuits. One particularly interesting avenue involves leveraging Matchgate circuits, which are known to be classically simulatable, incorporate into the training set. 
Another intriguing question concerns the broader applicability of the CPDR framework's Clifford perturbation circuit training set. Examining its effectiveness in other learning-based approaches, such as replacing the linear model with Random Forests or Multi-Layer Perceptrons, suggests a valuable research direction. 
Furthermore, exploring alternative input features, such as those beyond multiple noise-level expectation values and noisy expectation values of circuits with inserted Pauli gates, while training mitigation models using the CPDR dataset, is another meaningful avenue.
These investigations may have potential to improve the accuracy of QEM.

\textit{Note on Ref.~\cite{lerch2024efficient}:} Shortly before submitting this manuscript to the arXiv, we became aware of the work~\cite{lerch2024efficient}, which presents a similar result to Theorem~\ref{thm:truncation_error} in this manuscript. These two works were developed independently and are tailored to different applications.

\begin{acknowledgments}

We thank Zhenyu Chen for valuable discussions. 
R.Z., S.C., Z.W. and Z.L. were supported, in part, by the Beijing Natural Science Foundation under Grant No. Z220002; R.Z. and Z.W. were supported, in part, by the National Natural Science Foundation of China under Grant Nos. 62272259 and 62332009. Y.S., F.W. and Z.L. were supported, in part, by the BMSTC and ACZSP under Grant No.~Z221100002722017.
S.C. was supported by the National Science Foundation of China (Grant No. 12004205).
Z.L. was supported by NKPs (Grant No. 2020YFA0713000).

\end{acknowledgments}


\bibliographystyle{apsrev}
\bibliography{ref}

\clearpage
\widetext

\appendix
\section*{Supplementary Material}
\renewcommand{\thesection}{\Roman{section}}
\renewcommand{\appendixname}{Supplement Material}


\section{Clifford Perturbation Approximation}\label{ap:sec:CPA}

We consider the expectation value
\begin{equation}\label{eq:exp_val_m}
  \begin{aligned}
      \langle O \rangle =  \tr{{U}_L(\theta_L)  \cdots {U}_1(\theta_1) \rho  {U}_1(\theta_1)^{\dagger} \cdots  {U}_L(\theta_L)^{\dagger} O } 
  \end{aligned}
\end{equation}
of a n-qubit Pauli operator $O\in\{\mathbb{I},X,Y,Z\}^{\otimes n}$. Here, $\rho$ is the $n$-qubit initial state, and the sequence ${U}_L(\theta_L)  \cdots {U}_1(\theta_1)$ represents the quantum circuit. Each unitary ${U}_i(\theta_i):=C_i\exp{-i \theta_i P_i / 2} $ consists of a rotation on Pauli operator $P_i\in\{\mathbb{I},X,Y,Z\}^{\otimes n}$ with angle $\theta_i$ and a Clifford operator $C_i$.

In Heisenberg picture, the expectation value in Eq.~(\ref{eq:exp_val_m}) can be rewritten as $\langle O \rangle = \tr{\rho \tilde{O}}$, where $\tilde{O} = {U}_1(\theta_1)^{\dagger} \cdots  {U}_L(\theta_L)^{\dagger} O {U}_L(\theta_L)  \cdots {U}_1(\theta_1)$ is the Heisenberg-evolved observable. 
Thus, evaluating the expectation value is equivalent to applying the Heisenberg evolution to the observable $O$ and then measuring the initial state $\rho$. 

Specifically, an Heisenberg evolution of a Pauli operator $P'$ under ${U}(\theta):=C\exp{-i \theta P / 2}$ is given by ${U}(\theta)^{\dagger} P' {U}(\theta)= \exp{i \theta P / 2} C^\dagger P C\exp{-i \theta P / 2}$.
The Clifford $C$ maps Pauli operators to Pauli operators, $C^\dagger P' C=Q$ is a Pauli operator, and $Q$ can be efficiently computed. The application of a Pauli rotation to a Pauli operator yields:
\begin{equation}
e^{i \theta P / 2} Q e^{-i \theta P / 2} = 
\begin{cases}
Q, & [P, Q] = 0, \\
\cos(\theta) Q + i \sin(\theta) P Q & \{P, Q\} = 0.
\end{cases} \label{eq:evolution}
\end{equation}

When $Q$ anticommutes with $P$, the Pauli operator $Q$ is transformed into a linear combination of Pauli operators. We call those different Pauli operators as different path. Then $\cos(\theta) Q + i \sin(\theta) P Q$ can be viewed as sum of two paths.
To formalise evolution, we introduce the concept of Pauli paths.

A Pauli path is a sequence $s=(s_0,\cdots,s_L)\in \bm{P}^{L+1}_n$, where $\bm{P}_n=\{\sfrac{\mathbb{I}}{\sqrt{2}},\sfrac{X}{\sqrt{2}},\sfrac{Y}{\sqrt{2}},\sfrac{Z}{\sqrt{2}}\}^{\otimes n} $ represents the set of all normalized $n$-qubit Pauli words. 
The Pauli path $s$ is associated with a Heisenberg-evolved path. 
For example, at beginning, the $s_L$ is set to the normalized observable $\sfrac{O}{(\sqrt{2})^n}$ to record the evolution starts from $O$.
Then $s_{L-1}$ records the potential Pauli operator after the evolution of $s_L$ under the gate $U_L(\theta_L)=C_L\exp{-i \theta_i P_L / 2}$.
That is, if $P_L$ commutes with $C_L^\dagger s_L C_L$, then $s_{L-1}$ is $\sfrac{C_L^\dagger s_L C_L}{(\sqrt{2})^n}$.
Otherwise, in anticommute case, the $s_{L-1}$ has two potential choices, $\sfrac{C_L^\dagger s_L C_L}{(\sqrt{2})^n}$ or $\sfrac{P_L C_L^\dagger s_L C_L}{(\sqrt{2})^n}$ to record the different evolution paths. 
Iterating this process untill $s_0$ is obtained, the whole path $s$ describes a evolution path of the observable $O$. 
The phase in the above Pauli product is ignored, and it is records in the following path weight:
\begin{equation}
  \begin{aligned}\label{eq:f}
    f(s,\bm{\theta})=&\Tr{Os_L}\left(\prod_{i=1}^{L}\Tr{s_iU_i(\theta_i) s_{i-1}U_i(\theta_i)^\dagger}\right).
  \end{aligned}
\end{equation}
If $s$ is a path generated by the evolution of $O$, then $f(s,\bm{\theta})$ is the coefficient compose of $\sin$ and $\cos$, while if $s$ is not a path generated by the evolution of $O$, then $f(s,\bm{\theta})=0$.

Using the Pauli weight, Heisenberg-evolved observable can be expressed as 
\begin{equation}
  \begin{aligned}
    \tilde{O} &= {U}_1(\theta_1)^{\dagger} \cdots  {U}_L(\theta_L)^{\dagger} O {U}_L(\theta_L)  \cdots {U}_1(\theta_1)\\
    &=\sum_{s_0\in \bm{P}_n} \tr{{U}_1(\theta_1)^{\dagger} \cdots  {U}_L(\theta_L)^{\dagger} O {U}_L(\theta_L)  \cdots {U}_1(\theta_1)s_0} s_0\\
    &=\sum_{s_0\in \bm{P}_n} \tr{{U}_2(\theta_2)^{\dagger} \cdots  {U}_L(\theta_L)^{\dagger} O {U}_L(\theta_L)  \cdots {U}_1(\theta_1)s_0{U}_1(\theta_1)^{\dagger}} s_0\\
    &=\sum_{s_0,s_1\in \bm{P}_n} \tr{{U}_2(\theta_2)^{\dagger} \cdots  {U}_L(\theta_L)^{\dagger} O {U}_L(\theta_L)  \cdots {U}_2(\theta_2)s_1} \tr{s_1{U}_1(\theta_1)s_0{U}_1(\theta_1)^{\dagger}} s_0\\
    &\vdots\\
    &=\sum_{s\in \bm{P}^{L+1}_n}\Tr{Os_L}\left(\prod_{i=1}^{L}\Tr{s_iU_i(\theta_i) s_{i-1}U_i(\theta_i)^\dagger}\right) s_0\\
    &=\sum_{s\in \bm{P}^{L+1}_n} f(s,\bm{\theta}) s_0,
  \end{aligned}
\end{equation}
where the second and fourth equality holds because the $\bm{P}_n=\{\sfrac{\mathbb{I}}{\sqrt{2}},\sfrac{X}{\sqrt{2}},\sfrac{Y}{\sqrt{2}},\sfrac{Z}{\sqrt{2}}\}^{\otimes n} $ form a orthonormal basis.
Then the expectation value can be expressed as
\begin{equation}
  \begin{aligned}
    \langle O \rangle &= \tr{\rho \tilde{O}}\\
    &=\sum_{s\in \bm{P}^{L+1}_n} f(s,\bm{\theta}) \tr{\rho s_0}.
  \end{aligned}
\end{equation}

The number of Pauli path reflects the computational cost of evaluating the expectation value.
The worst-case scaling of this method is $2^L$, which is attained only if all $P_i$ anticommute with $s_i$. 
To effectively evaluate the expectation value, we truncate the number of sin terms in $f(s,\bm{\theta})$. We denote the set of Pauli paths with at most $M$ $\sin{\theta}$ terms in their Pauli weight as $\bm{P}^{L+1}_n(M)$, where $M$ is a positive integer.

Notation $\langle O \rangle ^{(M)}$ denotes the expectation value of the truncated Pauli weight set $\bm{P}^{L+1}_n(M)$, which can be expressed as
\begin{equation}
  \begin{aligned}
      \langle O \rangle ^{(M)} =\sum_{s\in \bm{P}^{L+1}_n(M)} f(s,\bm{\theta}) \tr{\rho s_0}.
  \end{aligned}
\end{equation}

\subsection{Error Analysis}

When rotation angles are small, the $\sin{\theta}$ terms in the Pauli weight are small, thus the truncations in $\sin{\theta}$ terms will provide a good approximation to the expectation value. To show this, we introduce the small angle space $\Theta=\{\bm{\theta}\mid \abs{\theta_i}\leq \theta_* ,i=1,\cdots,L \}$, and analysis the truncation error in the small angle space. 

\begin{lemma}
For any distinct Pauli paths $s,s^{\prime}\in \bm{P}^{L+1}_n$, and the small angle space $\Theta=\{\bm{\theta}\mid \abs{\theta_i}\leq \theta_* ,i=1,\cdots,L \}$, we have 
\begin{equation}\label{ap:eq:E_cross_equals_0}
\mathbb{E}_{\Theta}f(s,\bm{\theta})f(s',\bm{\theta})=0.
\end{equation}
\end{lemma}
\begin{proof}
    This lemma can be proved by following the evolution process. By Eq.~\eqref{eq:f}, we have
    \begin{equation}
f(s,\bm{\theta})f(s',\bm{\theta})=\Tr{Os_L}\Tr{Os'_L}\left(\prod_{i=1}^{L}\Tr{s_iU_i(\theta_i) s_{i-1}U_i(\theta_i)^\dagger}\Tr{s'_iU_i(\theta_i) s'_{i-1}U_i(\theta_i)^\dagger}\right).
    \end{equation}
    First, there must be $s_L=s'_L$ since the initial observable is the same, if not, the factor $\Tr{Os_L}\Tr{Os'_L}$ is zero. 
    Then, we consider the evolution of $s_{L-1}$ and $s'_{L-1}$ under the operator $U_L(\theta_L)=C_L\exp{-i \theta_i P_L / 2}$. 

    If $P_L$ commutes with $C_L^\dagger s_L C_L$, by Eq.~\eqref{eq:evolution}, we have $U_L(\theta_L)^\dagger s_L U_L(\theta_L)=U_L(\theta_L)^\dagger s'_L U_L(\theta_L) = C_L^\dagger s_L C_L$, then $s_{L-1}=s'_{L-1}$, otherwise the factor $\Tr{s_LU_L(\theta_L) s_{L-1}U_L(\theta_L)^\dagger}\Tr{s'_LU_L(\theta_L) s'_{L-1}U_L(\theta_L)^\dagger}$ is zero.
    
    If $P_L$ anticommutes with $C_L^\dagger s_L C_L$, by Eq.~\eqref{eq:evolution}, $s_{L-1}$ and $s'_{L-1}$ can be $C_L^\dagger s_L C_L$ or $P_L C_L^\dagger s_L C_L$, otherwise $\Tr{s_LU_L(\theta_L) s_{L-1}U_L(\theta_L)^\dagger}\Tr{s'_LU_L(\theta_L) s'_{L-1}U_L(\theta_L)^\dagger}=0$. If $s_{L-1}\neq s'_{L-1}$, means that there is factor $\cos{\theta_L}\sin{\theta_L}$ in the product $f(s,\bm{\theta})f(s',\bm{\theta})$, which is zero because
    \begin{equation}
        \int_{-\theta_*}^{\theta_*}\cos{\theta_L}\sin{\theta_L}\mathrm{d}\theta_L=0.
    \end{equation}
    Thus, we must have $s_{L-1}=s'_{L-1}$. Repeating this process, we have $s=s'$ otherwise $\mathbb{E}_{\Theta}f(s,\bm{\theta})f(s',\bm{\theta})=0$, which completes the proof.
\end{proof}

\begin{lemma}
    If observable $O$ is a Pauli operator, the mean square error between $\langle O \rangle$ and $\langle O \rangle ^{(M)}$ in the small angle space $\Theta$ is upper bounded by
    \begin{equation}
        \mathbb{E}_{\Theta}\left(\langle O \rangle - \langle O \rangle ^{(M)}\right)^2 \leq  \left(1+c\right)^L - \left(1+c\right)^M,
    \end{equation}
    where $c=\frac{\theta_*^2}{3}+\order{\theta_*^4}$ is a constant.
\end{lemma}

\begin{proof}
    By the definition of the expectation value, we have
    \begin{equation}
        \begin{aligned}
        \mathbb{E}_{\Theta}\left(\langle O \rangle - \langle O \rangle ^{(M)}\right)^2 &= \mathbb{E}_{\Theta}\left(\sum_{s\in\bm{P}^{L+1}_n(L)/\bm{P}^{L+1}_n(M)} f(s,\bm{\theta}) \tr{\rho s_0}\right)^2\\
        &=\sum_{s,s'\in\bm{P}^{L+1}_n(L)/\bm{P}^{L+1}_n(M)} \mathbb{E}_{\Theta}\left(f(s,\bm{\theta})f(s',\bm{\theta}) \tr{\rho s_0}\tr{\rho s'_0}\right)\\
        &=\sum_{s\in\bm{P}^{L+1}_n(L)/\bm{P}^{L+1}_n(M)} \mathbb{E}_{\Theta}\left(f(s,\bm{\theta}) \tr{\rho s_0}\right)^2,
      \end{aligned}
    \end{equation}
    where the last equality holds by Eq.~\eqref{ap:eq:E_cross_equals_0}.

    By Hölder's inequality $\abs{\tr{A^\dagger B}}\leq \norm{A}_p\norm{B}_q$ for $p,q>0$ satisfies $\frac{1}{p}+\frac{1}{q}=1$, we have
    \begin{equation}
        \begin{aligned}
        (f(s,\bm{\theta})\tr{\rho s_0})^2=&\Tr{Os_L}^2\tr{\rho s_0}^2\left(\prod_{i=1}^{L}\Tr{s_iU_i(\theta_i) s_{i-1}U_i(\theta_i)^\dagger}\right)^2\\
        \leq & (\norm{O}_2\norm{s_L}_2 \norm{\rho}_1\norm{s_0}_\infty)^2  \left(\prod_{i=1}^{L}\Tr{s_iU_i(\theta_i) s_{i-1}U_i(\theta_i)^\dagger}\right)^2\\
        =& \left(\prod_{i=1}^{L}\Tr{s_iU_i(\theta_i) s_{i-1}U_i(\theta_i)^\dagger}\right)^2,
        \end{aligned}
    \end{equation}
    where the last equality holds because $\norm{O}_2=\sqrt{\tr{OO^\dagger}}=\sqrt{2^n}$, $\norm{\rho}_1=\tr{\rho}=1$, $\norm{s_L}_2=\sqrt{\tr{s_L^\dagger s_L}}=1$, $\norm{s_0}_\infty=\left(\frac{1}{\sqrt{2}}\right)^2$.

    The set $\bm{P}^{L+1}_n(l)/\bm{P}^{L+1}_n(l-1)$ denotes the set of Pauli paths with exact $l$ $\sin{\theta}$ terms in their Pauli weight. The number size of set $\bm{P}^{L+1}_n(l)/\bm{P}^{L+1}_n(l-1)$ is upper bounded by $\binom{L}{l}$, because there are totally $L$ gates. 

    If we denote constant $c$ as 
    \begin{equation}
        c=\mathbb{E}_{\abs{\theta}\leq \theta_*}\sin^2{\theta}=\frac{1}{2\theta_*} \int_{-\theta_*}^{\theta_*}(\sin^2{\theta})\mathrm{d}\theta= \frac{1}{2}-\frac{\sin{2\theta_*}}{2\theta_*}=\frac{ \theta_*^2}{3}+\order{\theta_*^4}.
    \end{equation}

    Then for any $s\in\bm{P}^{L+1}_n(l)/\bm{P}^{L+1}_n(l-1)$, we have
    \begin{equation}
        \begin{aligned}
        \mathbb{E}_{\Theta}\left(f(s,\bm{\theta}) \tr{\rho s_0}\right)^2 &\leq \mathbb{E}_{\Theta}\left(\prod_{i=1}^{L}\Tr{s_iU_i(\theta_i) s_{i-1}U_i(\theta_i)^\dagger}\right)^2\\
        &= \prod_{i=1}^{L}\mathbb{E}_{\abs{\theta_i}\leq \theta_*}\Tr{s_iU_i(\theta_i) s_{i-1}U_i(\theta_i)^\dagger}^2\\
        & \leq c^{l},
        \end{aligned}
    \end{equation}
    where the inequality holds by $\Tr{s_iU_i(\theta_i) s_{i-1}U_i(\theta_i)^\dagger}^2$ equals to $\cos^2{\theta_i}$ or $\sin^2{\theta_i}$ or $1$, and the expectation value of $\mathbb{E}_{\abs{\theta_i}\leq \theta_*}\cos^2{\theta_i}\leq 1$ and $\mathbb{E}_{\abs{\theta_i}\leq \theta_*}\sin^2{\theta_i}=c$.

    Therefore, we have
    \begin{equation}
    \begin{aligned}
        \sum_{s\in\bm{P}^{L+1}_n(L)/\bm{P}^{L+1}_n(M)} \mathbb{E}_{\Theta}\left(f(s,\bm{\theta}) \tr{\rho s_0}\right)^2 &\leq \sum_{l=M+1}^{L} \sum_{\{s\in \bm{P}^{L+1}_n(l)/\bm{P}^{L+1}_n(l-1) \}} \mathbb{E}_{\Theta}\left(f(s,\bm{\theta}) \tr{\rho s_0}\right)^2\\
        & \leq \sum_{l=M+1}^{L} \sum_{\{s\in \bm{P}^{L+1}_n(l)/\bm{P}^{L+1}_n(l-1) \}} \left(\frac{\theta_*^2}{3}\right)^l \\
        & \leq \sum_{l=M+1}^{L} \binom{L}{l} c^l\\
        & \leq  \left(1+c\right)^L-\left(1+c\right)^M.\\
      \end{aligned}
\end{equation}

\end{proof}

The worst-case truncation error is upper bounded the following lemma.
\begin{lemma}
    For Pauli observable $O$, the difference between $\langle O \rangle$ and $\langle O \rangle ^{(M)}$ in the small angle space $\Theta$ is upper bounded by
    \begin{equation}
        \abs{\langle O \rangle - \langle O \rangle ^{(M)}} \leq  \left(1+\sin{\theta_*}\right)^L - \left(1+\sin{\theta_*}\right)^M.
    \end{equation}
\end{lemma}

\begin{proof}
    By the definition of the expectation value, we have
    \begin{equation}
        \begin{aligned}
            \abs{\langle O \rangle - \langle O \rangle ^{(M)}} &= \abs{\sum_{s\in\bm{P}^{L+1}_n(L)/\bm{P}^{L+1}_n(M)} f(s,\bm{\theta}) \tr{\rho s_0}} \\
            &\leq \sum_{s\in\bm{P}^{L+1}_n(L)/\bm{P}^{L+1}_n(M)} \abs{f(s,\bm{\theta}) \tr{\rho s_0}}\\
            &= \sum_{l=M+1}^{L} \sum_{\{s\in \bm{P}^{L+1}_n(l)/\bm{P}^{L+1}_n(l-1) \}} \abs{f(s,\bm{\theta}) \tr{\rho s_0}}\\
            &\leq \sum_{l=M+1}^{L} \sum_{\{s\in \bm{P}^{L+1}_n(l)/\bm{P}^{L+1}_n(l-1) \}} \sin^l{\theta_*} \\
            &\leq \sum_{l=M+1}^{L} \binom{L}{l} \sin^l{\theta_*}\\
            &\leq  \left(1+\sin{\theta_*}\right)^L - \left(1+\sin{\theta_*}\right)^M.
        \end{aligned}
    \end{equation}
\end{proof}

For any $\delta>0$, if $0<c<\frac{\ln{1+\frac{\delta}{2}}}{L-M}\leq \frac{\ln{2}}{M}$, then we have:
\begin{equation}
  \begin{aligned}
     \left(1+c\right)^{L}-\left(1+c\right)^{M}&= \left(1+c\right)^{M} \left(\left(1+c\right)^{L-M}-1 \right)\\
    &\leq 2\left(1+\frac{\delta}{2}-1\right)=\delta.
  \end{aligned}
\end{equation}
where the last inequality is due to $(1+x)^m\leq e^{mx}$. There $(1+c)^M\leq e^{cM}\leq 2$ and $\left(1+c\right)^{L-M}\leq e^{c(L-M)}\leq 1+\frac{\delta}{2}$.

Thus we have the following theorem. 
\begin{theorem}
    For any given $\delta>0$, if $\frac{\ln{1+\frac{\delta}{2}}}{L-M}\leq \frac{\ln{2}}{M}$, to ensure the difference between the expectation value and the $M$ truncation $\sin$ approximation less than $\delta$ in worst case, we have $\theta_*=\frac{\ln{1+\frac{\delta}{2}}}{L-M}=\order{\frac{1}{L}}$. In average case, to make the mean square error less than $\delta$, we have $\theta_*=\sqrt{\frac{3\ln{1+\frac{\delta}{2}}}{L-M}}=\order{\frac{1}{\sqrt{L}}}$.
\end{theorem}

\subsection{Computational Cost}

In this section we are going to bound the computational cost of estimating $\langle O \rangle ^{(M)}$, while the input srate $\rho=\sum_{a,b}\rho_{a,b}\ketbra{a}{b}$ is a sparse matrice with $\mathrm{Poly}(n)$ non-zero terms and the observable $O$ is a linear combination of $\mathrm{Poly}(n)$ Pauli terms. The analysis is splitted into three step: Firstly, we calculate the size of $\abs{\bm{P}^{L+1}_n(M)}$ when $O$ is a single Pauli operator, then the size for general $O$ is multiplied a $\mathrm{Poly}(n)$ factor. Second, we estimate the computational cost for calculate the contribution $f(s,\bm{\theta}) \tr{\rho s_0}$ of a single Pauli path $s$.

When $L\geq M \geq 2$, the size of the truncated Pauli path set $\bm{P}^{L+1}_n(M)$ is upper bounded by $L^M$, for single Pauli observable $O$. This is because in $f(s,\bm{\theta})$, for any $i\in\{1,\cdots,L\}$, $\sin{\theta_i}$ can appear at most once in the Pauli weight, if $\sin{\theta_i}$ appears then $\cos{\theta_i}$ will not appear.
Thus the number of Pauli paths with exact $l$ $\sin{\theta}$ terms in their Pauli weight no more than $\binom{L}{l}$. And we have
\begin{equation}
  \begin{aligned}
    \abs{\bm{P}^{L+1}_n(M)}&\leq \sum_{l=0}^{M}\binom{L}{l}.
  \end{aligned}
\end{equation}
To proof the right side of the above inequality is upper bounded by $L^M$, we use the method of induction. When $M=2$ and $L\geq 2$, we have $\binom{L}{0}+\binom{L}{1}+\binom{L}{2}=1+L+L(L-1)/2\leq L^2$. Assume that the inequality holds for $M=k$, then we have $\sum_{l=0}^{k+1}\binom{L}{l} \leq L^k +\binom{L}{k+1}$. Since $\frac{L}{(k+1)!}\leq L-1$, we have $\binom{L}{k+1}\leq \frac{L^{k+1}}{(k+1)!} \leq L^k(L-1)=L^{k+1}-L^k$, then we have $\sum_{l=0}^{k+1}\binom{L}{l} \leq L^k +\binom{L}{k+1}\leq L^k+L^{k+1}-L^k=L^{k+1}$. Thus the inequality holds for $M=k+1$. This completes
\begin{equation}
    \begin{aligned}
        \abs{\bm{P}^{L+1}_n(M)}&\leq \sum_{l=0}^{M}\binom{L}{l}\leq L^M.
    \end{aligned}
\end{equation}

Thus, when $O$ is a linear combination of $\mathrm{Poly}(n)$ Pauli terms, the size $\abs{\bm{P}^{L+1}_n(M)}\leq \mathrm{Poly}(n) L^M $. For path weight $f(s,\bm{\theta})=\Tr{Os_L}\left(\prod_{i=1}^{L}\Tr{s_iU_i(\theta_i) s_{i-1}U_i(\theta_i)^\dagger}\right)$, each term can be effectively calculating by Clifford simulation~\cite{gottesman1998heisenberg} and Eq.~\eqref{eq:evolution}, with computational cost $\order{n}$. The total computational cost for a single $f(s,\bm{\theta})$ is $\order{nL}$.

For $\tr{\rho s_0}=\sum_{a,b}\rho_{a,b} \bra{b}s_0\ket{a}$, each terms can be calculated by read the $a$-row $b$-cow element in $s_0$, with cost $\order{n}$. Therefore, the total cost is $(\order{nL}+\order{n})\mathrm{Poly}(n) L^M=\order{\mathrm{Poly}(n) L^{M+1}}$.

\section{Details of Numerical simulation in Sec. \ref{sec:Numerical} B}\label{sec:CPDR}

In Sec. \ref{sec:Numerical} B, we compare the CPDR-ZNE and CPDR-PEC with the learning-based PEC and ZNE protocols through numerical simulations. Here, we provide the details of the our numerical simulation. 

\subsection{Noise model}

As described in Sec.~\ref{sec:Numerical} B, we employ a gate-based noise model that closely approximates the noise characteristics of realistic quantum devices. For each gate, we introduce depolarizing noise following the gate operation and account for thermal relaxation during the gate's duration. The $n$-qubit depolarizing channel is
\begin{equation}
    \mathcal{D}(\rho) = (1-\lambda) \rho + \lambda\frac{I_{2^n}}{2^n}, 
\end{equation}
where $\rho$ is $n$ qubit state, and $I_{2^n}$ is $2^n \times 2^n$ identity matrix, $\lambda \in [0,1]$ is depolarizing noise intensity. And for single-qubit thermal relaxation noise, the channel can be represented as
\begin{equation}
\begin{gathered}
\epsilon(\rho)=\operatorname{tr}_1\left[\Lambda\left(\rho^T \otimes I\right)\right], \quad \Lambda=\left(\begin{array}{cccc}
\epsilon_{T_1} & 0 & 0 & \epsilon_{T_2} \\
0 & 1-\epsilon_{T_1} & 0 & 0 \\
0 & 0 & 0 & 0 \\
\epsilon_{T_2} & 0 & 0 & 1
\end{array}\right) \\
\text{where} \quad \epsilon_{T_1}=e^{-T_{\text{gate}} / T_1}, \quad \epsilon_{T_2}=e^{-T_{\text{gate}} / T_2}.
\end{gathered}
\end{equation}
$T_1$ and $T_2$ is relaxation time and dephasing time, respectively, and $t_{\text{gate}}$ is gate operation time.

We assume for each qubit, $T_1 = 100 \, \mu s$, $T_2 = 50 \, \mu s$. For each single-qubit gate, a thermal relaxation process occurs during the gate operation time $t_{\text{gate}} = 300 \, ns$, and followed by single-qubit depolarizing noise with an intensity of $\lambda_{\text{single}}=0.01$.
For two-qubit gates, the gate operation time is $t_{\text{gate}} = 800 \, ns$, and the following two-qubit depolarizing noise has an intensity of $\lambda_{\text{double}}=0.04$.
We assume that the noise in this model represents the base noise level, characterized by a constant $\lambda$.

In certain mitigation protocols, such as ZNE, it is necessary to scale the noise level. For instance, when the noise level is scaled from $\lambda$ to $G \lambda$, where $G$ is the scaling factor, the characteristics of depolarizing noise and thermal relaxation noise are adjusted as follows.
Scaling the depolarizing noise intensities \(\lambda_{\text{single}}\) and \(\lambda_{\text{double}}\) to \( G\lambda_{\text{single}} \) and \( G\lambda_{\text{double}} \), respectively, and scaling the gate operation time \( t_{\text{gate}} \) to \( Gt_{\text{gate}} \). Furthermore, to adapt to varying hardware capabilities, we also benchmark our CPDR protocols under different base noise rates to evaluate whether they can maintain their advantages. When scaling the base noise rate, we adopt the same numerical approach described above.

\subsection{Mitigation protocols}

Below, we outline the settings for each mitigation protocol used in the numerical simulations. 
We benchmark our CPDR protocols against ZNE and learning-based PEC in Trotterized time evolution simulations involving 8, 9, 10, and 20 qubits.
To ensure fairness, the number of shots for every noisy circuit in each protocol is set to \( 10^4 \). Additionally, the simulations include the entire shot process, incorporating the effects of static measurement errors.

For the ZNE protocol, noise circuit expectations are scaled with noise levels using the factors \( G = 1, 1.2, 1.6 \). We apply linear, quadratic, and exponential functions to extrapolate the ideal circuit expectations, as also done in Ref.~\cite{kim2023evidence}. Specifically, for the exponential function, we used the form \( b \exp{a \lambda} \), where \( a \) and \( b \) are the fitting coefficients for the exponential extrapolation.

For CPDR-ZNE, the training set consists of 144 pairs of $(\theta_h, \theta_J)$: 
$$
\mathcal{T} = \left\{(\theta_h, \theta_J)
| \theta_h = i\pi/120,
\theta_J = -j\pi/120, i,j = 0,1,2,3,4,5,54,55,56,57,58,59\right\}.
$$
We scale the noise by factors \( G = 1, 1.2, 1.6 \) for each circuit in the training set.

For CPDR-PEC and learning-based PEC, we select a set of operations to insert Pauli gates $\mathcal{G} = \{g_1,g_2,\cdots,g_{21}\}$, where $g_1$ represents no additional gate insertion, and $g_2,g_3,\cdots,g_{21}$ represent insert single Pauli X or Pauli Z in target circuits. The training set for the learning-based PEC includes 2048 Clifford circuits which is construct by replace the rotation gate in target circuit to Clifford gate randomly. And each Clifford circuit is inserted Pauli gate by operator $g_1, g_2, \cdots, g_{21}$. 
As for CPDR-PEC, the training set
includes circuit which rotation angle $(\theta_J, \theta_h) \in \mathcal{T}$, and insert Pauli gate by operator $g_1, \cdots, g_{21}$.

The noise circuits for the 8,9 and 10 qubit case are simulated using MindSpore Quantum~\cite{xu2024mindspore}. For the 20-qubit case, the noise circuits are simulated using the tensor network simulator proposed in Ref.~\cite{noh2020efficient} with bond
dimension is $120$. In both cases, the noiseless circuits for benchmark the accuracy of mitigation result are simulated in MindSpore Quantum, and the noiseless circuits in training set of CPDR-ZNE and CPDR-PEC are obtained by SPD with truncated number $13$.

\section{Additional numerical results of CPDR-ZNE and CPDR-PEC}\label{sec:ap:Numerical}

The results of comparison of different QEM protocols with different numbers of qubits and Trotter steps are illustrated as following.
Specifically, we displays the mean squared error (MSE) over $\theta_J$, defined as $\mathbb{E}_{\theta_J}\abs{\langle O \rangle-\widehat{\langle O \rangle}}^2$, for various $\theta_h$, and the MSE over $\theta_h$ for different $\theta_J$. 
To ensure more accurate calculation of the MSE, for each set of parameters \((\theta_h, \theta_J)\), we repeated each pair of \((\theta_h, \theta_J)\) with  every mitigation protocol 100 times.
As quantum devices continue to evolve and their base noise decrease, we evaluate our framework under reduced base noise scenarios.
We follow the same steps to reduce the base noise $\lambda$ as those used for scaling the noise in ZNE.
The results demonstrate that our CPDR framework remains highly competitive, even when applied to future, more advanced quantum devices.

For an 8-qubit, 4-step circuit, we evaluate the MSE of different mitigation protocols over \(\theta_J\) and \(\theta_h\) under various base noises \( 0.1\lambda, 0.2\lambda, 0.5\lambda \), and \(\lambda\) used as a reference for comparison. For ZNE with linear, quadratic and exponential fitting function, we denote as ``ZNE-linear'', ``ZNE-quad'' and ``ZNE-exp'', respectively. We also present the MSE of the noise expectation as a comparison, denoted as ``noise''.

\begin{center}
\begin{minipage}{0.99\textwidth}
  \includegraphics[width=\textwidth]{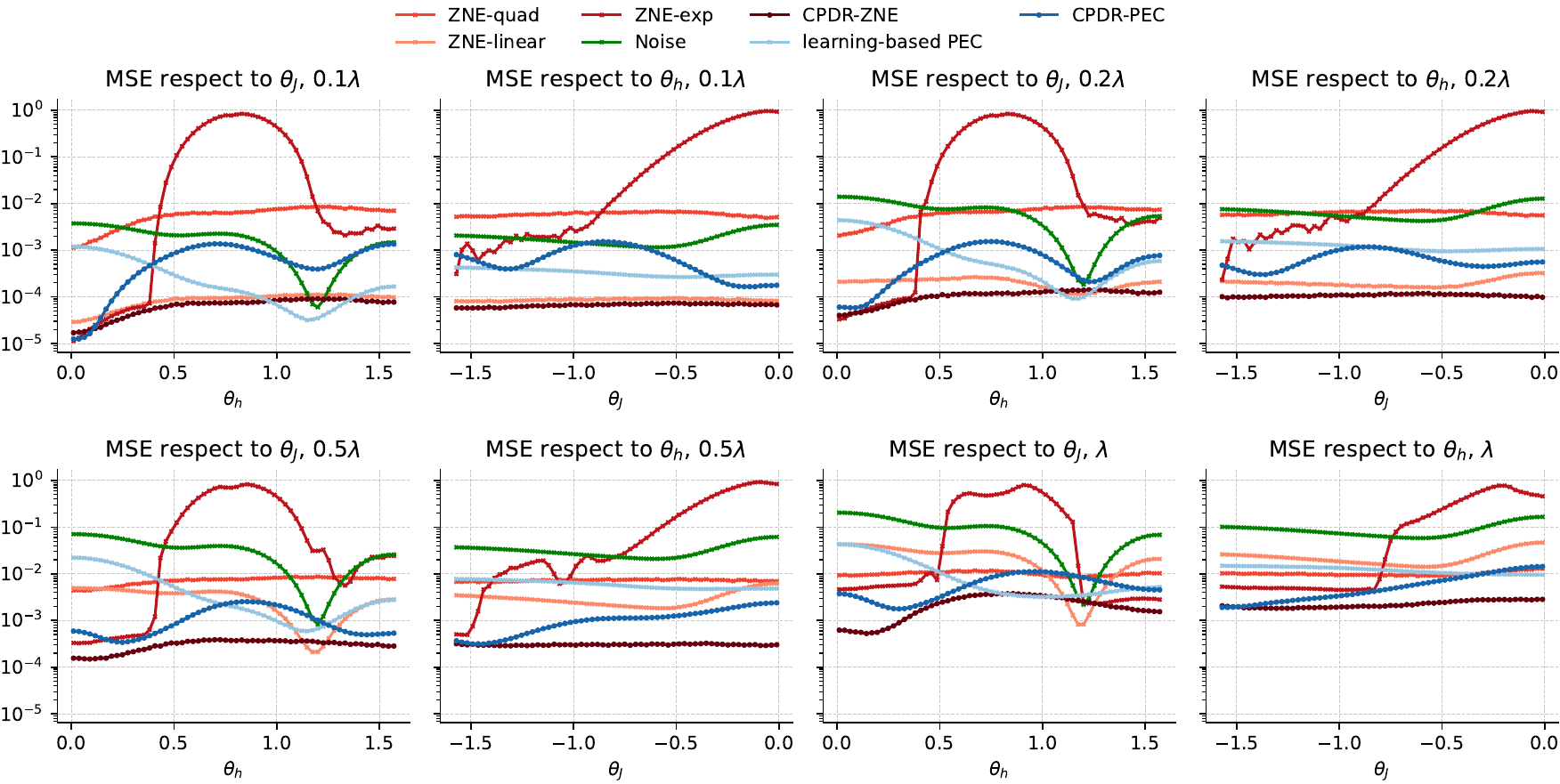}
\end{minipage}
\end{center}

For an 8-qubit, 6-step circuit, the results are as follows.

\begin{center}
\begin{minipage}{0.99\textwidth}
  \includegraphics[width=\textwidth]{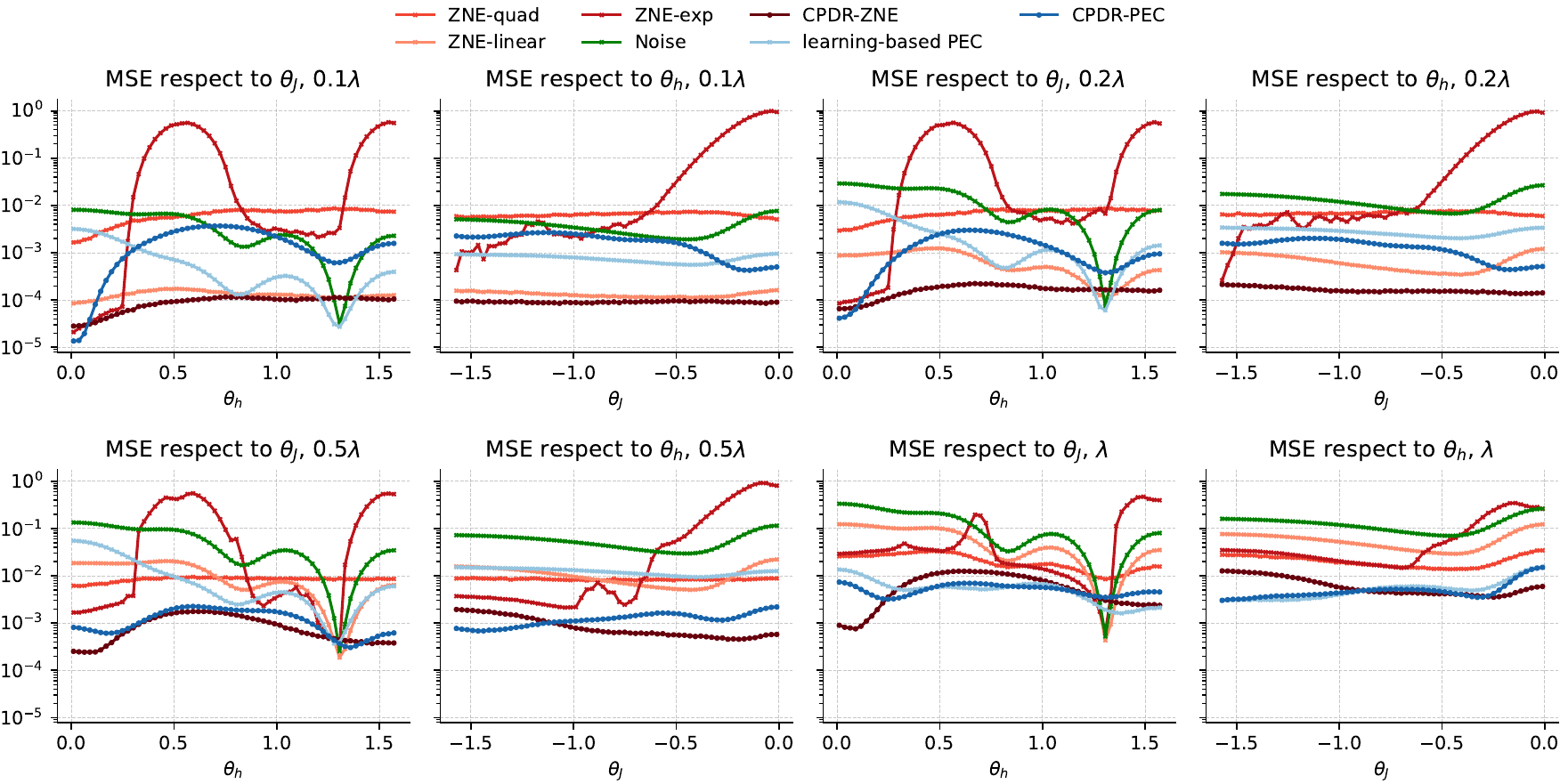}
\end{minipage}
\end{center}

For an 9-qubit, 5-step circuit, the results are as follows.

\begin{center}
\begin{minipage}{0.99\textwidth}
  \includegraphics[width=\textwidth]{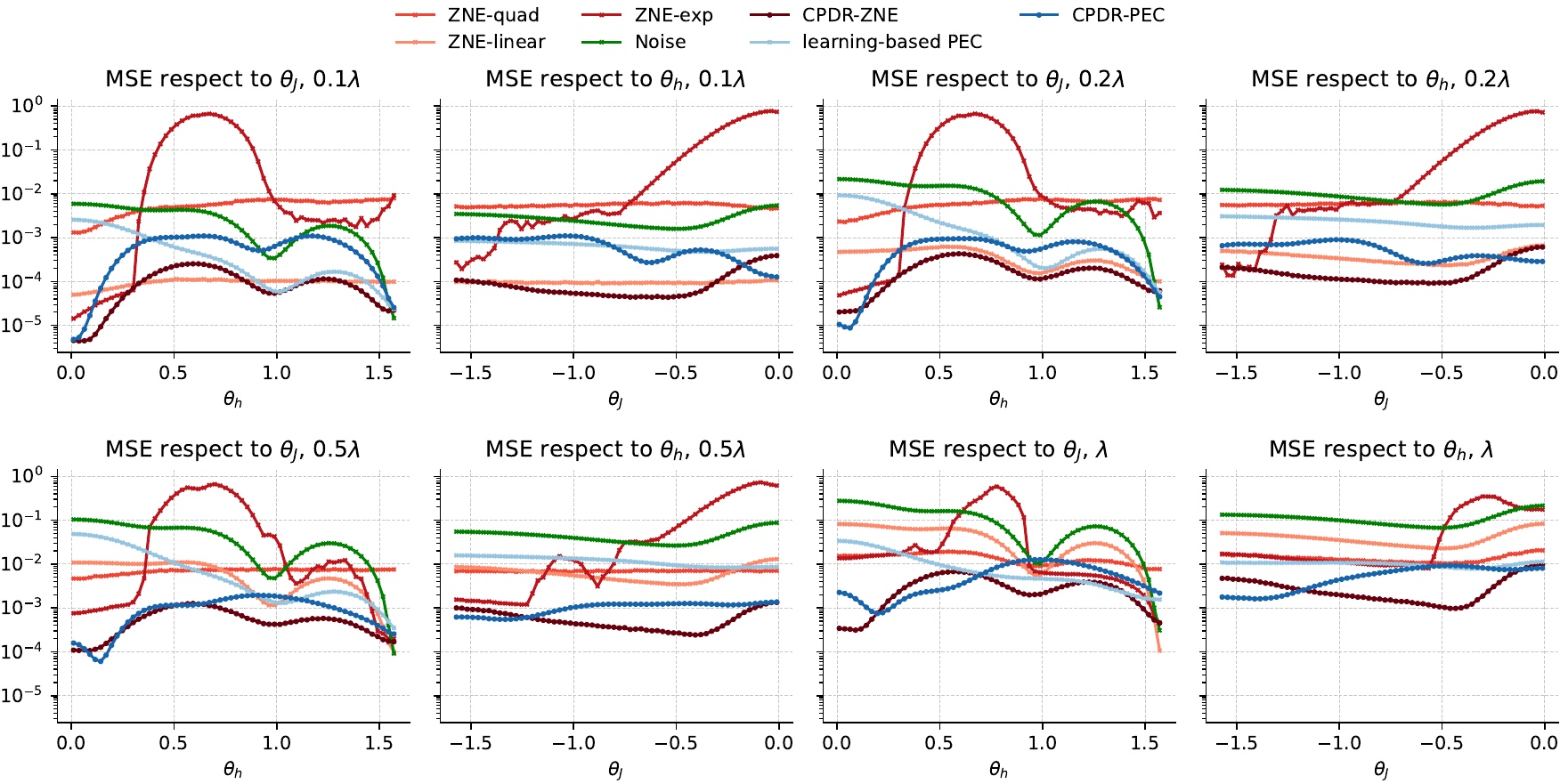}
\end{minipage}
\end{center}

For an 10-qubit, 5-step circuit, the results are as follows.

\begin{center}
\begin{minipage}{0.99\textwidth}
  \includegraphics[width=\textwidth]{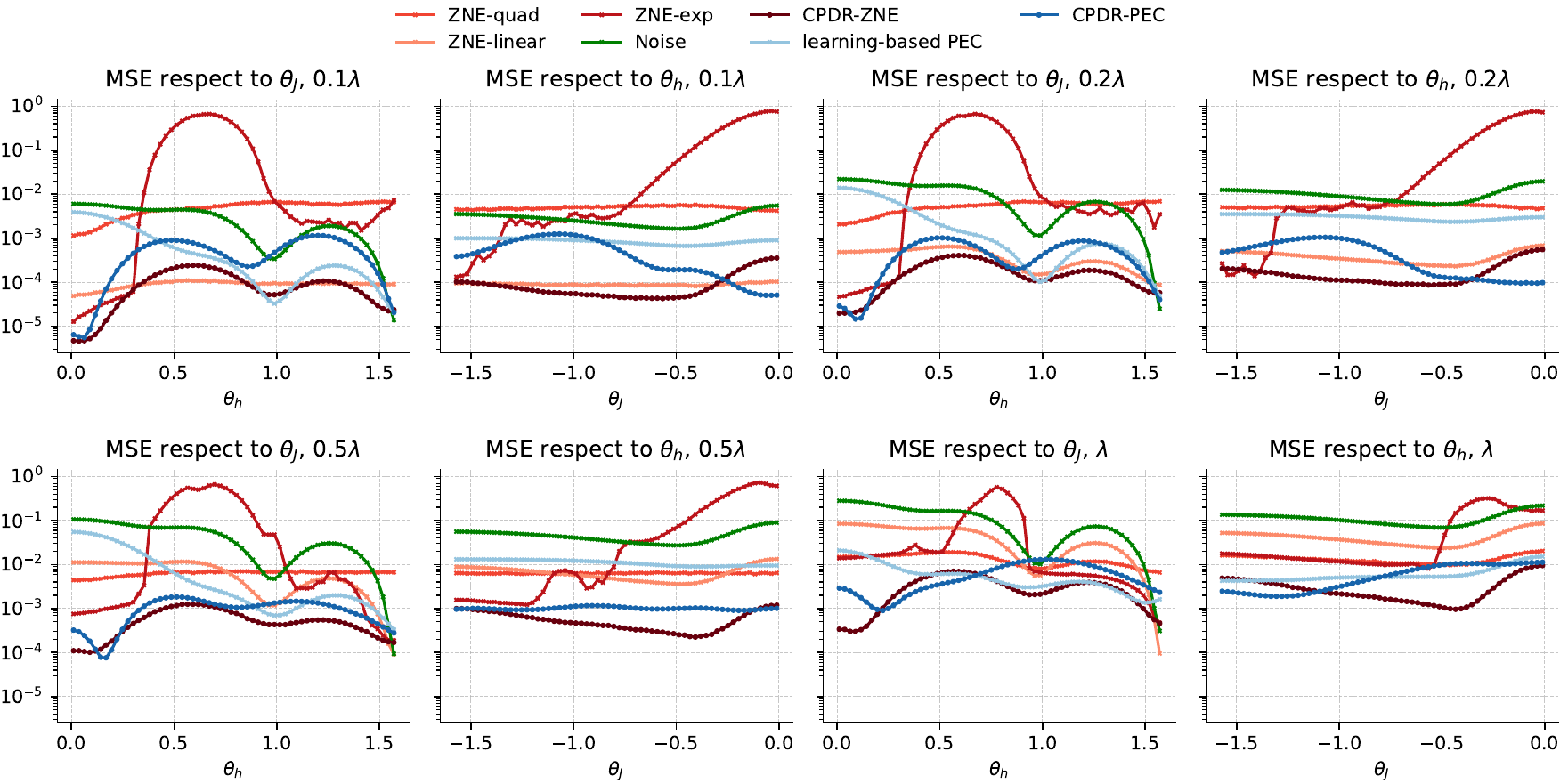}
\end{minipage}
\end{center}

In CPDR-ZNE and CPDR-PEC, the noiseless circuit expectation values in the training set are obtained through the SPD simulator. 
But there are inherent errors between the exact values and the SPD outputs.
Therefore, it is important to verify that the mitigation results are robust to the errors introduced by the SPD simulator.
To this end, we compare the MSEs of CPDR-ZNE and CPDR-PEC trained using both SPD simulator results and the exact values.
The results presented below show that the MSEs of the mitigation protocols trained with both approaches are nearly identical and thus demonstrate the robustness of our CPDR framework to the errors introduced by the SPD simulator.

For 8-qubit, 4-step with the base noise rate $\lambda,0.5\lambda,0.2\lambda$ and $0.1\lambda$, the MSE of CPDR-ZNE, CPDR-PEC trained with SDP and trained with exact value over $\theta_J$ and $\theta_h$ as follows.

\begin{center}
\begin{minipage}{0.99\textwidth}
  \includegraphics[width=\textwidth]{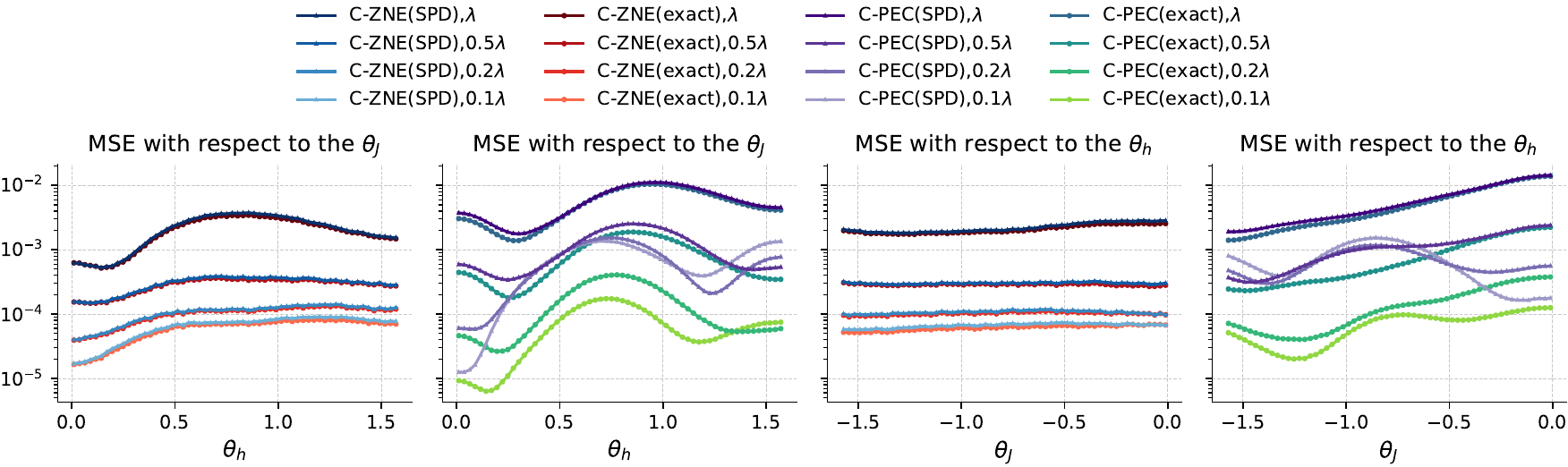}
\end{minipage}
\end{center}

For 8-qubit, 6-step, the results are as follows.

\begin{center}
\begin{minipage}{0.99\textwidth}
  \includegraphics[width=\textwidth]{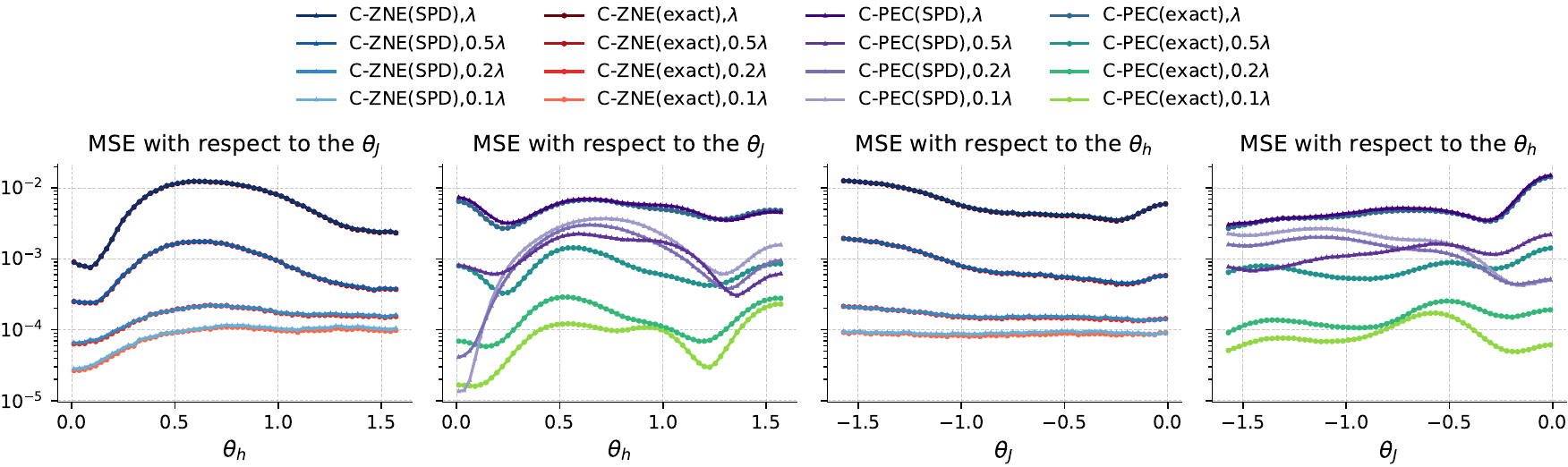}
\end{minipage}
\end{center}

For 9-qubit, 5-step, the results are as follows.

\begin{center}
\begin{minipage}{0.99\textwidth}
  \includegraphics[width=\textwidth]{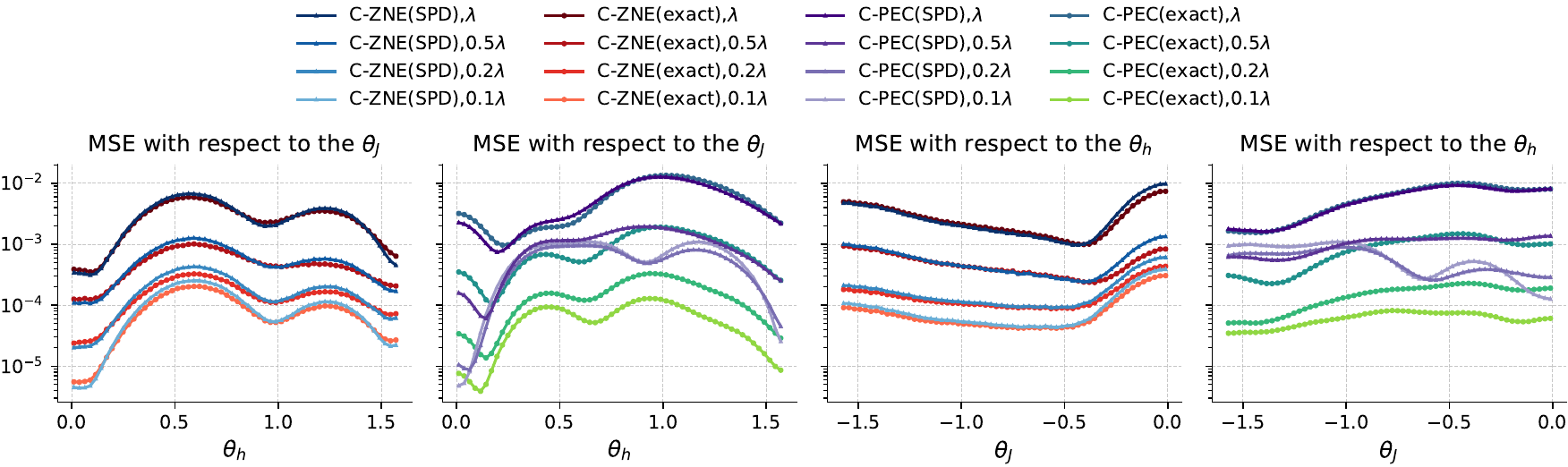}
\end{minipage}
\end{center}

For 10-qubit, 5-step, the results are as follows.

\begin{center}
\begin{minipage}{0.99\textwidth}
  \includegraphics[width=\textwidth]{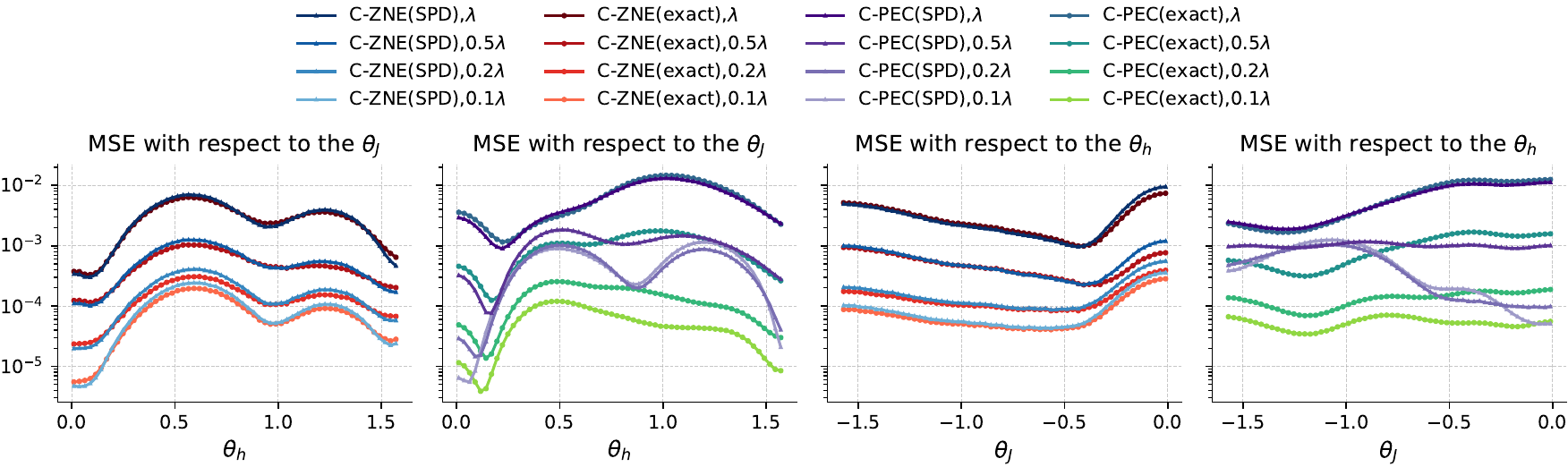}
\end{minipage}
\end{center}

\section{The quantum circuit on IBM's Eagle processor}\label{sec:ap:Numerical:IBM}

\begin{figure}[htbp] 
  \includegraphics[width=1\textwidth]{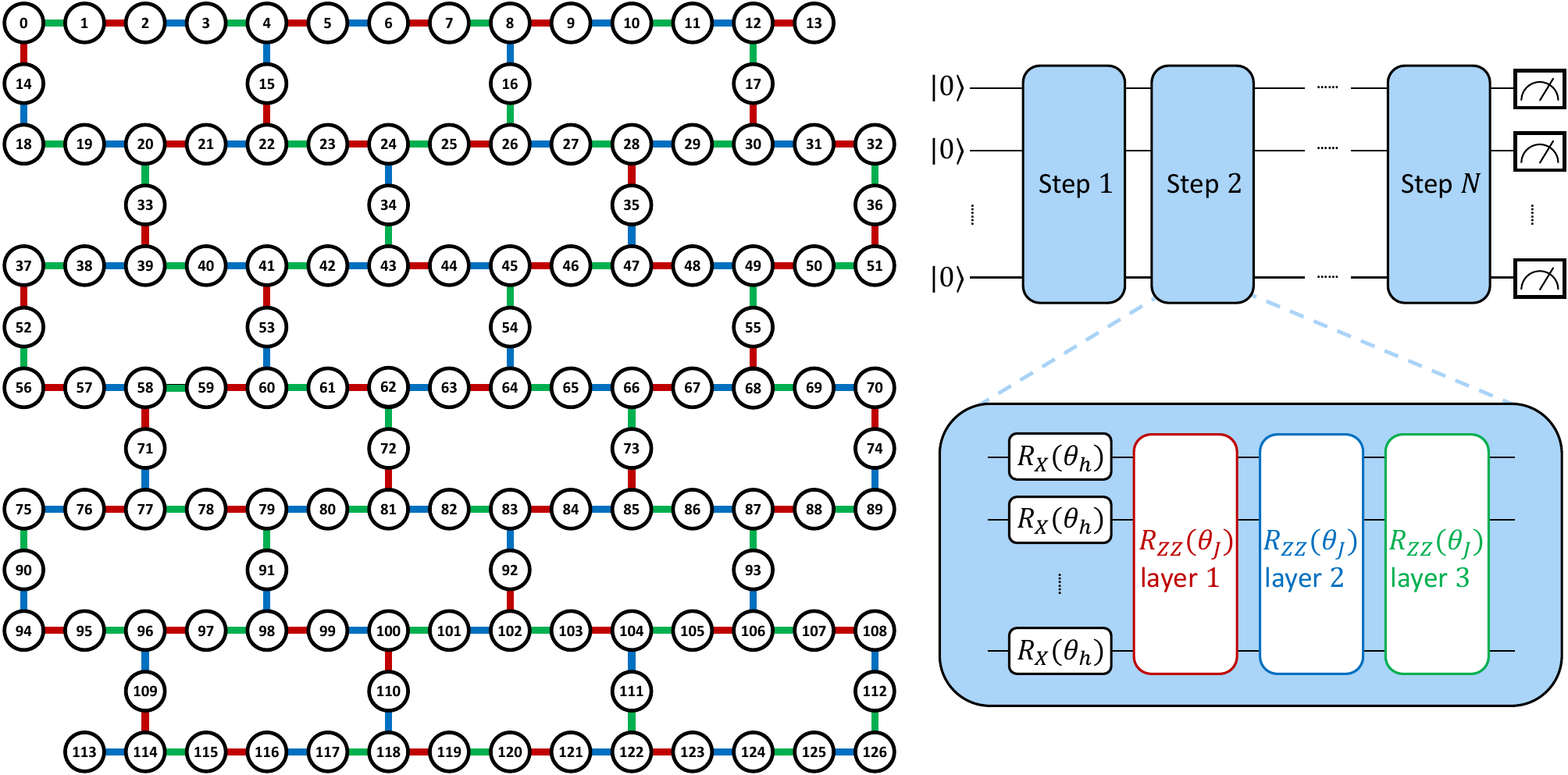}
  \captionsetup{justification=justified,singlelinecheck=false}
 \caption{\justifying
 The circuit for first-order Trotterized time evolution in IBM experiments comprises $N$ Trotter steps. Each step includes a layer of $R_X$ gates, all with a common rotation angle $\theta_h$ applied to every qubit, followed by three layers of $R_{ZZ}$ gates, each with a common rotation angle $\theta_J$. The qubits on which the $R_{ZZ}$ gates act are illustrated in the Eagle processor topology on the left and are marked with the same color.}\label{fig:IBM_Ansatz}
\end{figure}

In Sec. \ref{sec:Numerical} C, we examine four Trotterized time evolution circuits of Ising model in Fig.~\ref{fig:estimate_AQC}, which corresponding to Fig.~3(b,c), Fig.~4(a,b) of Ref.~\cite{kim2023evidence}. We compare our CPDR mitigation protocol with mitigation in \cite{kim2023evidence} with the same quantum hardware experiment data. Here we offer the detail of those circuits we used.

In Ref.~\cite{kim2023evidence}, IBM reported experiments performed on a 127-qubit Eagle processor.
The benchmark circuits were designed based on the Trotterized time evolution of a 2D transverse-field Ising model, tailored to align with the Eagle processor's topology.

The Hamiltonian that governs the system's time dynamics is
\begin{equation}
  H=-J\sum_{\langle i,j\rangle \in E} Z_iZ_j+h\sum_i X_i,
\end{equation}
where $J$ represents the coupling strength, $h$ is the transverse field strength, $E$ corresponds to the Eagle
processor’s topology.

The spin dynamics are simulated using first-order Trotterized time evolution of the Hamiltonian, expressed as
\begin{equation}\label{eq:time_evolution2}
  \begin{aligned}
    \ket{\psi(T)}&=\prod^N_{k=1} \exp{-i \frac{T}{N} \cdot H} \ket{\psi(0)}
\approx\prod^N_{k=1}\left(\prod_{\langle i, j\rangle \in E} \mathrm{e}^{i \frac{JT}{N} Z_i Z_j}\prod_i\mathrm{e}^{-i \frac{hT}{N} X_i }\right) \ket{\psi(0)} \\
  &=\prod^N_{k=1}\left(\prod_{\langle i, j\rangle \in E} \mathrm{R}_{Z_i Z_j}(-\frac{2 JT}{N})\prod_i \mathrm{R}_{X_i}(\frac{2Th}{N}) \right)\ket{\psi(0)}, \\
  \end{aligned}
\end{equation}
where the total evolution time $T$ is divided into $N$ Trotter steps. 
The Trotterized time evolution is realized by the ansatz depicted in Fig.~\ref{fig:IBM_Ansatz}. Each step consists of one layer of $R_X$ gates followed by three layers of $R_{ZZ}$ gates.

All the 127 qubits are initialized to $\ket{0}$. For simplicity, IBM set $\theta_J = -\frac{2JT}{N} = -\frac{\pi}{2}$ and varied $\theta_h =  \frac{2Th}{N}$ within the range $[0, \frac{\pi}{2}]$.

In Sec. \ref{sec:Numerical} C, four circuits in Fig.~\ref{fig:estimate_AQC}(a), (b), (c) and (d) are performed using the following observables and Trotter steps: 
\begin{enumerate} 
\item Fig.~\ref{fig:estimate_AQC}(a): The observable is a weight-10 operator $\left\langle X_{13,29,31} Y_{9,30} Z_{8,12,17,28,32} \right\rangle$. The circuit has $N=5$ Trotter steps.
\item Fig.~\ref{fig:estimate_AQC}(b): The observable is a weight-17 operator $\left\langle X_{37,41,52,56,57,58,62,79} Y_{75} Z_{38,40,42,63,72,80,90,91} \right\rangle$. The circuit has $N=5$ Trotter steps. 
\item Fig.~\ref{fig:estimate_AQC}(c): An additional layer of $R_x(\theta_h)$ gates is applied before measurement. The observable is a weight-17 operator $\left\langle X_{37,41,52,56,57,58,62,79} Y_{38,40,42,63,72,80,90,91} Z_{75} \right\rangle$. The circuit has $N=5$ Trotter steps. 
\item Fig.~\ref{fig:estimate_AQC}(d): The observable is $\langle Z_{62} \rangle$, a single-site magnetization for qubit 62. The circuit has $N=20$ Trotter steps.
\end{enumerate}

In the training process of CPDR-ZNE, we obtain the map between the noise circuit expectations and the ideal expectations by solving Eq.~\eqref{eq:L2_loss_clifford_perturbation}, with the regularization parameter $\alpha = 2\times 10^{-5}$. The noiseless expectation references in training set are obtained by SPD, the truncated number is set as $7,3,2,10$  for the four circuits, respectively.

\end{document}